\documentclass[submission,copyright,creativecommons]{eptcs}
\providecommand{\event}{ICE'18}
 
\usepackage{breakurl}

\usepackage{hyperref}
\usepackage{amsfonts}
\usepackage{amsmath}
\usepackage{amssymb}
\usepackage{amsthm}
\usepackage{amsbsy}
\usepackage{enumerate}
\usepackage{stackrel}
\usepackage{bm}
\def\colorPtp{\color{blue}}
\def\colorNode{\color{cyan}}
\def\colorOp{\color{OliveGreen}}
\def\colorMsg{\color{BrickRed}}

\usepackage{amsmath}
\usepackage{amssymb}
\usepackage{mfirstuc}
\newcommand{\aG}{\mathsf{G}}
\newcommand{\gatedistancein}{3pt}
\newcommand{\gatedistanceinand}{2pt}
\newcommand{\gname}[1][i]{{\colorNode{\scriptstyle\textsf{#1}}}}
\usepackage{xifthen}        
\usepackage{xstring}
\usepackage{xargs}
\usepackage[usenames,dvipsnames,svgnames,table]{xcolor} 
\usepackage{tikz}

\usetikzlibrary{
  arrows,snakes,shapes,automata,backgrounds,petri,positioning,fit,calc,
  decorations.markings,decorations.text,shadows,fadings,patterns,
  decorations.pathreplacing,chains,spy
}

\tikzset{
  src/.style={draw,circle,fill=white,
    minimum size=2mm,
    inner sep=0pt
  },
  sink/.style={draw,circle,double,fill=white,
    minimum size=1.5mm,
    inner sep=0pt
  },
  node/.style={draw,circle,fill=black,
    minimum size=2mm,
    inner sep=0pt
  },
  source/.style={draw,circle,fill=white,
    minimum size=3mm,
    inner sep=0pt
  },
  sink/.style={draw,circle,double,fill=white,
    minimum size=3mm,
    inner sep=0pt
  },
  block/.style = {rectangle, draw=gray, align=center, fill=orange!25, rounded corners=0.1cm,
    minimum size=5mm, inner sep=2pt},
  prenode/.style = {minimum size=9pt,inner sep=2pt, font=\Large},
  bblock/.style = {rectangle, draw=blue!50, opacity=.5, line width=1pt, align=center, fill=white, rounded corners=0.1cm,
    minimum size=7mm, inner sep=2pt},
  prenode/.style = {minimum size=9pt,inner sep=2pt, font=\Large},
  agate/.style={draw, rectangle,
    minimum size=3mm,
    inner sep=0pt,
    fill=orange!25,
    postaction={path picture={%
        \draw[red]
        ([yshift=\gatedistanceinand]path picture bounding box.south) --
        ([yshift=-\gatedistanceinand]path picture bounding box.north) ;}}
  },
  ogate/.style = {
    diamond, draw, fill=orange!25,
    minimum size=4mm,
    inner sep=0pt,
    postaction={path picture={%
        \draw[red]
        ([yshift=\gatedistancein]path picture bounding box.south) -- ([yshift=-\gatedistancein]path picture bounding box.north)
        ([xshift=-\gatedistancein]path picture bounding box.east) -- ([xshift=\gatedistancein]path picture bounding box.west)
        ;}}},
  altogate/.style = {
    diamond, draw,
    minimum size=4mm,
    inner sep=0pt,
    postaction={path picture={%
        \draw
        ([yshift=\gatedistancein]path picture bounding box.south) -- ([yshift=-\gatedistancein]path picture bounding box.north)
        ([xshift=-\gatedistancein]path picture bounding box.east) -- ([xshift=\gatedistancein]path picture bounding box.west)
        ;}}},
  altgate/.style={draw, rectangle,
    minimum size=3mm,
    inner sep=0pt,
    postaction={path picture={%
        \draw
        ([yshift=\gatedistanceinand]path picture bounding box.south) --
        ([yshift=-\gatedistanceinand]path picture bounding box.north) ;}}},
  anygate/.style = {circle, draw, fill=white,
    minimum size=4mm,
    inner sep=0pt,
    postaction={path picture={%
        \draw[black]
        ([xshift=-\gatedistancein,yshift=\gatedistancein]path picture bounding box.south east) --
        ([xshift=\gatedistancein,yshift=-\gatedistancein]path picture bounding box.north west)
        ([xshift=-\gatedistancein,yshift=-\gatedistancein]path picture bounding box.north east) --
        ([xshift=\gatedistancein,yshift=\gatedistancein]path picture bounding box.south west)
        ;}}
  },
  smallglobal/.style={
        node distance=1cm and 0.8cm, semithick, scale=0.8, every node/.style={transform shape}
  },
  elli/.style = {draw,densely dotted,-},
  line/.style = {draw,->, rounded corners=0.07cm,>=latex},
  nline/.style = {draw,semithick, ->},
  pline/.style = {draw,->,>=latex},
  node distance=1cm and 0.7cm,
  baseline=(current  bounding  box.center),
  local/.style={rectangle, draw, fill=\fillcolor, drop shadow,
    text centered, rounded corners, minimum height=5em
  },
  bigar/.style={
    draw,very thick, ->
  },
  process/.style={rectangle, draw=gray, fill=\fillcolor, drop shadow,
    text centered, minimum height=5em,text=gray
  },
  choreo/.style={rectangle, draw, fill=\fillcolor, drop shadow,
    text centered, rounded corners, minimum height=5em
  },
  mycfsm/.style={
        font=\footnotesize,
        initial where=above,
        ->,>=stealth,auto, node distance=1cm and 1cm,
        scale=1, every node/.style={transform shape},
        every state/.style=inner sep=2pt,
        baseline=(current  bounding  box.center)
  },
  machinecloud/.style={
    cloud, cloud puffs=10, cloud ignores aspect, minimum height=.1cm, minimum width=2cm, draw
  },
  fitting node/.style={
    inner sep=0pt,
    fill=none,
    draw=none,
    reset transform,
    fit={(\pgf@pathminx,\pgf@pathminy) (\pgf@pathmaxx,\pgf@pathmaxy)}
  },
  mypetri/.style={
    font=\footnotesize,
    baseline=(current  bounding  box.center)
  },
  silentrans/.style = {rectangle, draw=black, align=center, fill=black,
    minimum height=1pt,
    minimum width=15pt,
    inner sep=1.5pt
  },
  reset transform/.code={\pgftransformreset},
  tmtape/.style={draw,minimum size=1.2cm}
}

\newcommand{\p}{\ptp}
\newcommand{\q}{{\ptp[B]}}
\newcommand{\msg}[1][m]{\mathsf{\colorMsg{#1}}}
\newcommand{\ifempty}[3]{%
  \ifthenelse{\isempty{#1}}{#2}{#3}%
}
\newcommandx{\nmerge}[2][1={i},2={},usedefault=@]{
  \ifempty{#2}{
    \ifempty{#1}{\mu}{-\gname[{#1}]}
  }{-{#2}}
}
\newcommandx{\gint}[4][1=i,2=\ptp,3=\msg,4=\q,usedefault=@]{
  \ptp[{#2}] {\colorOp \xrightarrow{\scriptscriptstyle\gname[#1]}} \ptp[{#4}] \colon {\msg[{#3}]}
}
\newcommandx{\gout}[4][1=\gname,2=\ptp,3=\msg,4={\ptp[C]},usedefault=@]{
  \achan[{#2}][{#4}] {\colorOp {\colorOp{!}}} {\msg[{#3}]}
}
\newcommandx{\gin}[4][1=\gname,2=\ptp,3=\msg,4={\ptp[C]},usedefault=@]{
  \achan[{#2}][{#4}] {\colorOp {\colorOp{?}}} {\msg[{#3}]}
}
\newcommandx{\gseq}[3][1=i,2={\aG},3={\aG'},usedefault=@]{
  \gnode[{#1}][{#2} \gseqop {#3}]
}
\newcommandx{\gpar}[3][1=i,2={\aG},3={\aG'},usedefault=@]{
  \gnode[{#1}][\ifempty{#1}{{#2} \gparop {#3}}{({#2} \gparop {#3})}]
}
\newcommandx{\gcho}[3][1=i,2={\aG},3={\aG'},usedefault=@]{
  \gnode[{#1}][\ifempty{#1}{{#2} \gchoop {#3}}{\big({#2} \gchoop {#3}\big)}]
}
\newcommandx{\gchov}[3][1=i,2={\aG},3={\aG'},usedefault=@]{
  \gnode[{#1}][\left(
  \begin{array}l
    \ifempty{#1}{{#2} \\ \gchoop \\ {#3}}{\!\!{#2} \\ \gchoop \\ {#3}}
  \end{array}\right)
  ]
}
\newcommandx{\grec}[3][1=i,2={\aG},3={\p},usedefault=@]{
  \gnode[{#1}][\ifempty{#1}{\grecop {#2} \grecopp {#3}}{\big(\grecop {#2} \grecopp {#3}\big)}]
}
\newcommand{\ptp}[1][A]{
  \ensuremath{\mathtt{\colorPtp{\capitalisewords{#1}}}}
}

\newcommandx{\mkint}[6][3=i,4=\p,5=\msg,6=\q,usedefault=@]{
  \node[bblock,{#1}] (#2) {$\gint[#3][#4][#5][#6]$};
}

\newcommand{\mkseq}[2]{\path[line] (#1) -- (#2);}

\newcommandx{\mkgraph}[3][1=.5cm]{
  \node[source,above = #1 of {#2}] (src#2) {};
  \node[sink,below  = #1 of {#3}] (sink#3) {};
  \path[line] (src#2) -- (#2);
  \path[line] (#3) -- (sink#3);
}

\newcommandx{\mkloop}[4][1=.5,2=1.5]{ 
  \node[ogate,above = #1 of {#3}] (entry#3) {};
  \pgfgetlastxy \xentry \yentry;
  \pgfmathtruncatemacro{\xentryrounded}{\xentry};
  \node[ogate,below  = #1 of {#4}] (exit#4) {};
  \pgfgetlastxy \xexit \yexit;
  \pgfmathtruncatemacro{\xexitrounded}{\xexit};
  \path[line] (entry#3) -- (#3);
  \path[line] (#4) -- (exit#4);
  \pgfmathsetmacro\tmpdiff{abs(\xentryrounded - \xexitrounded)}
  \path[line] (exit#4) -|  ($(exit#4)+(\tmpdiff,0)+(#2,0)$) |- (entry#3);
}

\newcommandx{\mklooptwo}[4][1=.5,2=1.5]{ 
  \node[ogate,above = #1 of {#3}] (entry#3) {};
  \pgfgetlastxy \xentry \yentry;
  \pgfmathtruncatemacro{\xentryrounded}{\xentry};
  \path (#4);
   \pgfgetlastxy \xexit \yexit;
  \pgfmathtruncatemacro{\xexitrounded}{\xexit};
  \path[line] (entry#3) -- (#3);
  \pgfmathsetmacro\tmpdiff{abs(\xentryrounded - \xexitrounded)}
  \path[line] (#4) -|  ($(#4)+(\tmpdiff,0)+(#2,0)$) |- (entry#3);
}

\newcommandx{\mklooptwobelow}[4][1=.5,2=1.5]{ 
  \node[ogate,above = #1 of {#3}] (entry#3) {};
  \pgfgetlastxy \xentry \yentry;
  \pgfmathtruncatemacro{\xentryrounded}{\xentry};
  \path (#4);
   \pgfgetlastxy \xexit \yexit;
  \pgfmathtruncatemacro{\xexitrounded}{\xexit};
  \path[line] (entry#3) -- (#3);
  \pgfmathsetmacro\tmpdiff{abs(\xentryrounded - \xexitrounded)}
  \path[line] (#4)  |- ($(#4)+(0,-0.5)$) -|  ($(#4)+(\tmpdiff,0)+(#2,0)$) |- (entry#3);
}

\newcommandx{\mklooponetwo}[4][1=.5,2=1.5]{ 
  \path (#3);
  \pgfgetlastxy \xentry \yentry;
  \pgfmathtruncatemacro{\xentryrounded}{\xentry};
  \path (#4);
   \pgfgetlastxy \xexit \yexit;
  \pgfmathtruncatemacro{\xexitrounded}{\xexit};
  \path[line] (entry#3) -- (#3);
  \pgfmathsetmacro\tmpdiff{abs(\xentryrounded - \xexitrounded)}
  \path[line] (#4) -|  ($(#4)+(\tmpdiff,0)+(#2,0)$) |- (entry#3);
}

\newcommandx{\mkfork}[4][2=gatenode,3=i,4=.6,usedefault=@]{
  \mkgatebegin{#1}[{\gname[#3]}][agate][#4]{#2}
}

\newcommandx{\mkbranch}[4][2=gatenode,3=i,4=.6,usedefault=@]{
  \mkgatebegin{#1}[{\gname[#3]}][ogate][#4]{#2}
}

\newcommandx{\mkgatebegin}[5][2={},3=ogate,4=.5]{
  \coordinate (gatecord) at (0,0);
  \foreach \n [count=\i] in {#1}{
    \pgfgetlastxy \xc \yc;
    \path (\n);
    \pgfgetlastxy \xn \yn;
    \coordinate (gatecord) at ($(gatecord) + (\xn,0)$);
    \coordinate (gatecord) at ($1/\i*(gatecord)$);
    \ifdim \yn < \yc
    \node (max) at (0,\yc) {};
    \else
    \node (max) at (0,\yn) {};
    \fi
  }
  \coordinate (gatecord) at ($(gatecord) + (0,#4) + (max)$);
  \node[#3,label={below:$#2$}] (#5) at (gatecord) {};
  \pgfgetlastxy{\xgate}{\ygate};
  \pgfmathtruncatemacro{\xgateround}{\xgate};
  \StrCount{#1,}{,}[\l] 
  \ifnum \l < 2 {\errmessage{#1 argument should be a comma-separated list of lenght >= 2}}
  \else{
    \foreach \n in {#1}{
      \path (\n);
      \pgfgetlastxy{\xnode}{\ynode};
      \pgfmathtruncatemacro{\xnround}{\xnode};
      \pgfmathsetmacro\tmpdiff{abs(\xnround - \xgateround)}
      \ifdim \tmpdiff pt > 1 pt \path[line] (#5) -| (\n);
      \else
        \path[line] (#5) -- (\n);
      \fi
    }
  }
  \fi
}

\newcommandx{\mkmerge}[4][2=gatenode,3=i,4=0,usedefault=@]{\mkgateend{#1}[{\ifempty{#3}{}{\nmerge[#3]}}][ogate][#4]{#2}}

\newcommandx{\mkjoin}[4][2=gatenode,3=i,4=0,usedefault=@]{\mkgateend{#1}[{\ifempty{#3}{}{\nmerge[#3]}}][agate][#4]{#2}}

\newcommandx{\mkgateend}[5][2={},3=ogate,4=.5]{
  \coordinate (gatecord) at (0,0);
  \foreach \n [count=\i] in {#1}{
    \pgfgetlastxy \xc \yc;
    \path (\n);
    \pgfgetlastxy \xn \yn;
    \coordinate (gatecord) at ($(gatecord) + (\xn,0)$);
    \coordinate (gatecord) at ($1/\i*(gatecord)$);
    \ifdim \yn > \yc
    \node (min) at (0,\yc) {};
    \else
    \node (min) at (0,\yn) {};
    \fi
  }
  \coordinate (gatecord) at ($(gatecord) - (0,#4) + (min)$);
  \node[#3,label={above:$#2$}] (#5) at (gatecord) {};
  \pgfgetlastxy{\xgate}{\ygate};
  \pgfmathtruncatemacro{\xgateround}{\xgate};
  \StrCount{#1,}{,}[\l] 
  \ifnum \l < 2 {\errmessage{#1 argument should be a comma-separated list of lenght >= 2}}
  \else{
    \foreach \n in {#1}{
      \path (\n);
      \pgfgetlastxy{\xnode}{\ynode};
      \pgfmathtruncatemacro{\xnround}{\xnode};
      \pgfmathsetmacro\tmpdiff{abs(\xnround - \xgateround)}
      \ifdim \tmpdiff pt > 1 pt \path[line] (\n) |- (#5);
      \else
        \path[line] (\n) -- (#5);
      \fi
    }
  }
  \fi
}

\newtheorem{definition}{Definition}[section]
\newtheorem{lemma}[definition]{Lemma}
\newtheorem{proposition}[definition]{Proposition}

\newtheorem{corollary}[definition]{Corollary}
\newtheorem{remark}[definition]{Remark}

\newtheorem{fact}[definition]{Fact}

\DeclareMathAlphabet{\mathpzc}{OT1}{pzc}{m}{it}

\newcommand{\Comment}[1]{ }

\def\Pred[#1]{~[\,#1\,]}

\newcommand{\projecton}[2]{#1\!\!\downharpoonright\!#2}

\newcommand{\edgelabel}[3]{{\tt #1}#2{\sf #3}}

\newcommand{\GG}{{\bf G}}
\newcommand{\gt}{$\mathcal{G}\!\!\mathcal{T}$}
\newcommand{\gtir}{\gt\!-\textsc{ir}}
\newcommand{\GTIR}[2]{{[#1]^{\!\langle#2\rangle}}}
\newcommand{\II}{{\tt I}}
\newcommand{\JJ}{{\tt J}}
\newcommand{\HH}{{\tt H}}
\newcommand{\KK}{{\tt K}}

\renewcommand{\AA}{\texttt{A}}
\newcommand{\BB}{\texttt{B}}
\newcommand{\roles}{\mathbf{P}}

\newcommand{\interfacecomp}{\!\leftrightarrow\!}
\newcommand{\components}{\mathcal{C}}

\newcommand{\connect}[2]{\stackrel{\hspace{-3pt}#1\!\leftrightarrow\!#2\!}{}}

\newcommand{\restrict}[2]{{#1}_{\mid_{\mathbf{#2}}}}
\newcommand{\restrictup}[2]{{#1}^{\mid{\mathbf{#2}}}}

\newcommand{\ttp}{\mathtt{p}}
\newcommand{\ttq}{\mathtt{q}}
\newcommand{\ttr}{\mathtt{r}}
\newcommand{\tts}{\mathtt{s}}

\newcommand{\mC}{{\not\mathit{C}}}

\newcommand{\lang}[1]{\mathcal{L}(#1)}
\newcommand{\gateway}[1]{\mathsf{gw}(#1)}

\newcommand{\elle}{\mathit{l}}

\newcommand{\Dual}[1]{\overline{#1}}

\newcommand{\Sem}[1]{[\hspace{-0.6mm}[ #1 ]\hspace{-0.6mm}]}

\newcommand{\Set}[1]{\{#1\}}

\renewcommand{\implies}{~\Longrightarrow~ }

\newcommand{\lts}[1]{\stackrel{#1}{\longrightarrow}}
\newcommand{\ltsone}[1]{\stackrel{\!#1}{\longrightarrow_{\!1}}}

\newcommand{\Act}{\mathit{ Act}}

\newcommand{\preGTIR}{\textit{pre-GTIR}}

\def\titlerunning{ Global Types for Open Systems}
\def\authorrunning{
Barbanera, de'Liguoro \& Hennicker
}
\title{Global Types for Open Systems
\footnote{The first two authors were partially supported
by the COST Action EUTYPES CA-15123
and by, respectively,
Project ``Chance'' of the University of Catania and
Project FORMS 2015 of the University of Torino.}
}

\author{Franco Barbanera
\institute{Dipartimento di Matematica e Informatica\\
University of Catania}
\email{barba@dmi.unict.it}
 \and 
Ugo de'Liguoro
\institute{Dipartimento di Informatica\\
University of Torino}
\email{ugo.deliguoro@unito.it}
\and 
Rolf Hennicker
\institute{Institute of Informatics,
LMU Munich}
\email{hennicker@ifi.lmu.de}
}

\begin{document}

\setlength{\abovedisplayskip}{6pt}
\setlength{\belowdisplayskip}{\abovedisplayskip}

\maketitle

\begin{abstract}
Global-type formalisms enable to describe the overall behaviour of distributed systems and at the same time to enforce  safety properties for communications between system components.
Our goal is that of amending a weakness of such formalisms:
the difficulty in describing {\em open} systems, i.e.\ systems which can be connected and interact with other open systems.
We parametrically extend, with the notion of {\em interface role} and {\em interface connection},
the syntax of global-type formalisms.
Semantically, global types with interface roles denote open systems of communicating finite state machines connected by means of 
{\em gateways} obtained from compatible interfaces.
We show that safety properties are preserved when
open systems are connected that way.
\end{abstract}

\section{Introduction}
\label{sect:Intro}

The intrinsic difficulties programmers have to face when developing and verifying distributed applications 
have been variously attacked by the theoretical computer science community with the aim of devising  formal systems
enabling
(1) to describe in a structured way the overall behavior of a system, and
(2)  to steer the implementation
of the system components, guaranteeing their compliance with the overall behaviour together with some relevant 
 properties of communications.

Several formalisms based on the notion of {\em global type} have been  proposed in the literature to pursue  such an aim
\cite{CHY07,CDP12,CDYP16}.  
The expressiveness of the investigated formalisms kept on increasing during the last decade, recently leading to
representations of global behaviours as graphs \cite{DY12,TY15,TG18}, where the local
end-point projections are interpreted by communicating finite state machines (CFSMs), a widely investigated
formalism for the description and the analysis of distributed systems \cite{BZ83}.
For systems of CFSMs, most of the relevant properties of communications are, in general, undecidable
\cite{CF05} or computationally hard. Instead, systems of CFSMs obtained by projecting 
the generalised
global types of~\cite{DY12} or the global graphs of~\cite{TY15,TG18} (more precisely those which satisfy a {\em well-formedness} condition)
are guaranteed to satisfy desired properties of communications like
 deadlock-freeness, that any sent message is eventually consumed
 or that each participant will eventually receive any message s/he is waiting for~\cite{DY12}.
   
The centralised viewpoint offered by the global type approaches makes them naturally suitable for describing
{\em closed} systems. This prevents a system described/developed by means of global types
to be looked at as a module that can be connected to other systems.
The description and analysis of {\em open} systems has been investigated, instead, in the context of CFSMs in~\cite{HB18,H16}
and for synchronous communication in the context of interface automata in~\cite{deAlfaro2001,deAlfaro2005}. 
In the present paper we address  the problem of
generalising the notion of {\em global type} in order to encompass  the description of {\em open systems}
and, in particular, open systems of CFSMs; so paving the way towards a fruitful interaction between the investigations on open systems carried out in automata theory and those on global types. \\

In our approach, an ``open global type'' -- that we dub ``global type with interface roles'' (GTIR) -- denotes
a number of connected open systems of CFSMs  where some participants (roles\footnote{We prefer to use the word {\em role} rather then {\em participant} since {\em interface role}
sounds more suitable for the present setting  than {\em interface participant.}}) are identified as interfaces rather than proper participants.
We have no necessity to stick to any particular global type formalism as a basis for our GTIRs,
as long as the local end-point behaviours of a global type $G$ can be interpreted as CFSMs.
So we introduce a {\em parametric} syntax which, given a global type formalism \gt,
extends its syntax by essentially
enabling to identify some roles as {\em interface roles} and to 
define a  composition of open global types, semantically interpreted by systems of CFSMs.
We call \gtir\ (\gt-with-\textsc{i}nterface-\textsc{r}oles) the so obtained formalism.

Syntactically, a GTIR is either a global type $G$ (formulated in \gt) together with a distinguished subset of the roles of $G$ declared as interface roles, or it is 
a composite expression where two GTIRs are composed via compatible interfaces.
The non-connected interface roles remain open after composition.  
The semantics of a GTIR is always a set of CFSMs. In the case of a basic GTIR,
i.e.\ a global type $G$ equipped with interface roles,
it is just the set of CFSMs obtained by projecting $G$ to its end-point CFSMs.
Those CFSMs interpreting interface roles model
the expected behaviour of an external environment of the open system.
Interface roles are \emph{compatible} if their CFSMs have no mixed states, are input and output deterministic
and if their languages are dual to each other. 
If a GTIR $\GG$ is a composite expression, composing GTIRs
$\GG_1$ and $\GG_2$ via compatible interface roles $\HH$ and $\KK$,
then the semantics of $\GG$ is the union of the two CFSM systems
denoted by $\GG_1$ and $\GG_2$ where the CFSMs
$M_{\HH}$ and $M_{\KK}$, interpreting
the interface roles $\HH$ and $\KK$, are replaced by appropriate gateway CFSMs
$\gateway{M_\HH,\KK}$ and $\gateway{M_\KK,\HH}$.
These gateways are
constructed by a simple algorithm out of $M_{\HH}$ and $M_{\KK}$.

A main objective of our work is to study the preservation of safety properties under composition. We consider three kinds of properties: deadlock-freedom, freedom of orphan messages and freedom of unspecified receptions following the definitions in~\cite{DY12} (which in turn follow definitions in~\cite{CF05}).
The main result of the present paper is that these
safety properties hold for the CFSM system $S$ denoted by a GTIR $\GG$
whenever they hold singularly for all subsystems $S_i$ obtained by the semantics of the 
global types $G_i$ that are used for the construction of $\GG$.
In particular, it has been shown in~\cite{DY12} that the safety properties are ensured whenever the  $G_i$'s are
{\em well-formed}  generalised global types

\noindent
{\em Overview.} In Section~\ref{sect:cfsm} the main definitions concerning CFSMs and systems of CFSMs are recalled, together with the definitions of safety properties. Syntax and semantics of global types with interface roles are introduced in Section~\ref{sec:gtir}. 
Section~\ref{sect:safetypreservation} studies the preservation of safety properties for systems connected via gateways and proves
our main result. 
Section~\ref{sect:conclusions} concludes
by pointing to related work and describing directions for further investigations.


\section{Systems of Communicating Finite State Machines}
\label{sect:cfsm}
In this section we recall (partly following \cite{CF05,DY12,TY15}) the definitions of communicating finite state machine (CFSM) and  systems of CFSMs.
Throughout the paper we assume given a countably infinite set  
$\roles_\mathfrak{U}$ of role (participant) names (ranged over by $\ttp,\ttq,\ttr,\tts,\tt{A},\tt{B},\HH,\II,\ldots$) and a countably infinite alphabet $\mathbb{A}_\mathfrak{U}$ (ranged over by $a,b,c,\ldots$)
of messages.\\

\begin{definition}[CFSM]\label{def:cfsm}
Let $\roles$  and $\mathbb{A}$ be finite subsets of $\roles_\mathfrak{U}$ and $\mathbb{A}_\mathfrak{U}$ respectively.
\begin{enumerate}[i)] 
\item
The set $C_\roles$ of {\em channels} over $\roles$ is defined by\\
\centerline{
$C_\roles=\Set{\ttp\ttq \mid \ttp,\ttq\in \roles, \ttp\neq\ttq}$
}
\item
The set $\mathit{Act}_{\roles,\mathbb{A}}$ of {\em actions}  over $\roles$ and $\mathbb{A}$ is defined by\\
\centerline{
$\textit{Act}_{\roles,\mathbb{A}} = C_\roles\times\Set{!,?}\times\mathbb{A}$
}

\item
A {\em communicating finite-state machine over} $\roles$ \emph{and} $\mathbb{A}$
is a finite transition system given by a tuple\\
\centerline{ $M=(Q,q_0,\mathbb{A},\delta)$ }
where $Q$ is a finite set of states, $q_0\in Q$ is the initial state, and
$\delta\subseteq Q\times\textit{Act}_{\roles,\mathbb{A}}\times Q$ is a set of transitions.
\end{enumerate}
\end{definition}

Notice that the above definition of a CFSM is generic w.r.t.\ the underlying sets
$\roles$ of roles and $\mathbb{A}$ of messages.
This is necessary,  since we shall not deal with a single system of CFSMs but with an arbitrary number of open systems that can be 
{\em composed}.
We shall write $C$ and $\mathit{Act}$ instead of $C_\roles$ and $\mathit{Act}_{\roles,\mathbb{A}}$ when no ambiguity can arise.
 We assume $\elle,\elle',\ldots$ to range over $\textit{Act}$;
$\varphi,\varphi',\ldots$ to range over $\textit{Act}^*$ (the set of finite words over $\textit{Act}$), and
$w,w',\ldots$ to range over $\mathbb{A}^*$ (the set of finite words over $\mathbb{A}$).
$\varepsilon\,(\notin \mathbb{A}\cup\textit{Act})$ denotes the empty word and $\mid v\mid$ the lenght of a word $v\in \textit{Act}^*\cup\mathbb{A}^*$.
The transitions of a CFSM are labelled by actions; a label $\tts\ttr!a$ represents
the asynchronous sending of message $a$ from machine $\tts$ to $\ttr$ through channel $\tts\ttr$ and, dually,
$\tts\ttr?a$ represents the reception (consumption) of $a$ by $\ttr$ from channel
$\tts\ttr$. 

We write $\lang{M}\subseteq\textit{Act}^*$ for
the language over $\textit{Act}$ accepted by the automaton corresponding
to machine $M$, where each state of $M$ is an accepting state. A state
$q\in Q$ with no outgoing transition is final; $q$ is a {\em sending} (resp. {\em receiving})
state if all its outgoing transitions are labelled with sending
(resp. receiving) actions; $q$ is a {\em mixed} state otherwise.

\vspace{2mm}
A CFSM $M = (Q,q_0,\mathbb{A},\delta)$ is:
\begin{enumerate}[a)]
\item
 {\em deterministic} if for all states $q\in Q$ and all actions $\elle$: 
$(q,\elle, q'), (q,\elle,q'')\in \delta$ imply $q'=q''$;
\item
{\em ?-deterministic} (resp. {\em !-deterministic}) if for all states $q\in Q$ and all actions
$(q,\ttr\tts?a, q'), (q,\ttp\ttq?a,q'')\in \delta$ (resp. $(q,\ttr\tts!a, q'), (q,\ttp\ttq!a,q'')\in \delta$) imply $q'=q''$;
\item
{\em ?!-deterministic} if it is both ?-deterministic and !-deterministic.
\end{enumerate}

The notion of ?!-deterministic machine is more demanding than in usual CFSM settings. It will be needed in order to guarantee safety-properties preservation when systems are connected. \\

\begin{definition}[Communicating systems and configurations]
Let $\roles$  and $\mathbb{A}$ be as in Def.~\ref{def:cfsm}.
\begin{enumerate}[i)]
\item
A {\em communicating system (CS)
over} $\roles$ \emph{and} $\mathbb{A}$ is a tuple $S= (M_\ttp)_{\ttp\in\roles}$
where\\
for each $\ttp\in \roles$,
$M_\ttp=(Q_\ttp,q_{0\ttp},\mathbb{A},\delta_\ttp)$ is a CFSM  over $\roles$ and $\mathbb{A}$.
\item
A {\em configuration} of a system $S$ is a pair $s = (\vec{q},\vec{w})$
where \\
-  $\vec{q}= (q_\ttp)_{\ttp\in\roles}$ with $q_\ttp \in Q_\ttp$;\\
-  $\vec{w}  = (w_{\ttp\ttq})_{\ttp\ttq\in C}$ with $w_{\ttp\ttq}\in\mathbb{A^*}$.\\
The component $\vec{q}$ is the control state of the system and $q_\ttp \in Q_\ttp$ is the local
state of machine $M_\ttp$. 
The component $\vec{w}$ represents the state of the channels of the system and $w_{\ttp\ttq} \in \mathbb{A}^*$ is the state of the channel for messages sent from $\ttp$ to $\ttq$. The initial configuration of $S$ is $s_0=  (\vec{q_0},\vec{\varepsilon})$
with $\vec{{q_0}} = (q_{0_\ttp})_{\ttp\in\roles}$.
\end{enumerate}
\end{definition}

\begin{definition}[Reachable configurations]
Let $S$ be a communicating system, and let $s= (\vec{q},\vec{w})$ and $s'= (\vec{q'},\vec{w'})$ 
be two configurations of $S$. 
Configuration $s'$ {\em is reachable from} $s$
{\em by firing  a transition} with action $\elle$, written $s\lts{\elle}s'$, if there is $a\in\mathbb{A}$
such that one of the following conditions holds:
\begin{enumerate}
\item
$\elle = \tts\ttr!a$ and $(q_\tts,\elle,q'_\tts)\in\delta_\tts$ and
\begin{enumerate}[a)]
\item
for all $\ttp\neq\tts: ~ q'_\ttp =  q_\ttp$  and
\item
$w'_{\tts\ttr} =  w_{\tts\ttr}\cdot a$ and for all $\ttp\ttq\neq\tts\ttr: ~ w'_{\ttp\ttq} =  w_{\ttp\ttq}$;
\end{enumerate}
\item 
$\elle = \tts\ttr?a$ and $(q_\ttr,\elle,q'_\ttr)\in\delta_\ttr$ and
\begin{enumerate}[a)]
\item
for all $\ttp\neq\ttr: ~ q'_\ttp =  q_\ttp$  and
\item
$w_{\tts\ttr} =  a\cdot w'_{\tts\ttr}$ and for all $\ttp\ttq\neq\tts\ttr: ~w'_{\ttp\ttq} =  w_{\ttp\ttq}$.
\end{enumerate}
We write $s\lts{}s'$ if there exists $\elle$ such that  $s\lts{\elle}s'$.\\ 
As usual, we denote the reflexive and transitive 
closure of $\lts{}$ by $\lts{}^*$
\item
The set of {\em reachable configurations} of S is $RS(S) = \Set{s \mid s_0\lts{}^* s}.$
\end{enumerate}
\end{definition}
\noindent
According to the last definition, communication happens via buffered channels following the FIFO principle.


\begin{definition}[Safety properties \cite{DY12,CF05}]\hfill\\
\label{def:safeness}
Let $S$ be a communicating system, and let $s= (\vec{q},\vec{w})$ be a configuration of $S$.
\begin{enumerate}[i)]
\item
\label{def:safeness-i}
$s$ is a {\em deadlock configuration} if 
$$\vec{w}=\vec{\varepsilon} 
~~\wedge~~ \forall \ttp\in\roles.~q_\ttp \text{ is a receiving state}
$$
i.e. all buffers are empty, but all machines are waiting for a message.\\
We say that $S$ is {\em deadlock-free} whenever, for any $s\in RS(S)$, $s$ is not a  deadlock configuration.
\item
$s$ is an {\em  orphan-message  configuration} if
$$(\forall \ttp\in\roles. ~ q_\ttp \text{ is final}) ~~\wedge~~ \vec{w}\neq \vec{\varepsilon}$$
i.e. each machine is in a final state, but there is still  at least one non-empty buffer.
\item
\label{def:safeness-ur}
$s$ is an {\em unspecified reception configuration}  if
$$\exists \ttr \in\roles. ~ q_\ttr \text{ is a receiving state } 
~~\wedge~~\forall\tts\in\roles.[~(q_\ttr,\tts\ttr?a,q'_\ttr)\in\delta_\ttr  \implies
(|w_{\tts\ttr}| > 0~~\wedge~~ w_{\tts\ttr}\not\in a\mathbb{A}^*)~ ]$$
i.e. there is a receiving  state $q_\ttr$ 
which is prevented from
receiving any message from any of its buffers.
(In other words, in each channel $\tts\ttr$ from which role $\ttr$ could consume there
is a message which cannot be received by $\ttr$ in state $q_\ttr$.)  
\item
$S$ is {\em safe} if,  for each $s\in RS(S)$, $s$
is neither a deadlock, nor an orphan-message, nor an unspecified reception
configuration.

\end{enumerate}
\end{definition}


The above definitions of safety properties are the same as in~\cite{DY12}.
They follow, for the notions of deadlock and unspecified reception, the definitions in~\cite{CF05}. The deadlock definition in~\cite{TY15} is slightly weaker, but coincides
with~\cite{DY12} if the local CFSMs have no final states. Still weaker definitions of deadlock are used in~\cite{TG18} and in~\cite{BZ83}.

\section{Global Types with Interface Roles}
\label{sec:gtir}

Our aim is the development of a formalism \gtir\ suitable for the composition of open systems which ensures preservation
of safety properties and hence is suitable for modular system construction.
The idea is that our approach should be usable for any global type formalism
\gt which satisfies the following assumptions: For each global type $G$ in \gt,

\begin{enumerate}
\item
there is associated a finite set of roles $\roles(G) \subset \roles_\mathfrak{U}$ and a finite set of actions
 $\mathbb{A}(G) \subset \mathbb{A}_\mathfrak{U}$,
 \item
 there is a {\em projection} function, denoted by $\projecton{\_}{\_}$,
 such that for any $\ttp\in\roles(G)$, $\projecton{G}{\ttp}$ is a CFSM
 over $\roles(G)$ and $\mathbb{A}(G)$.
\end{enumerate}

Global types are considered as syntactic objects while the projection function
yielding a communicating system  $(\projecton{G}{\ttp})_{\ttp\in\roles(G)}$  over $\roles(G)$ and $\mathbb{A}(G)$ is considered 
as the semantics of a global type $G$.
A \gt\ formalism satisfying the above requirements could  be, for instance, the formalism of generalised global types in~\cite{DY12}
or that of global graphs in~\cite{TY15,TG18}.
Under certain conditions on the form of a global type $G$, safety properties are guaranteed
for the system obtained by projecting all roles of $G$ to a CFSM.
For instance,  Theorem 3.1 in~\cite{DY12} states that a communicating system generated from the projections of a well-formed global type is safe in the sense of
Def.~\ref{def:safeness}.
Therefore, if we can assure that safety is preserved by composition, we get a safe system
whenever the underlying global types denote safe (sub)systems.

In order to look at global types as {\em open}, we shall identify some of their roles as {\em interface roles}.
An interface roles represents (part of) the expected communication behaviour of the environment. Interface roles are the basis to compose systems.
In this section, we introduce
Global Types  with Interface Roles (GTIRs)
and provide their syntax and semantics.
First we define the syntactic notion of a pre-GTIR.
Then we present our working example which already points out that
we need a semantic compatibility relation between interfaces for safe composition.
Interface compatibility and a few additional conditions must be respected
to get a proper GTIR. The semantics of a GTIR is then defined
in terms of a system of CFSMs obtained, in the base case, by the projections
of a global type with distinguished interface roles, and, in the composite case,
by composing open CFSM systems by means of suitable ``gateway'' CFSMs. 

\subsection{Pre-GTIRs}

A pre-GTIR is either just
a global type where some roles are declared as interface roles
or it is a syntactic expression composed from two pre-GTIRs by connecting
certain interface roles. The non-connected interface roles remain open.


\begin{definition}[pre-GTIR]
\label{def:pre-gtir}
The set of {\em pre-GTIR} expressions $\GTIR{{\GG}}{\mathbf{I}}$ with set of {\em interface roles} $\mathbf{I}$ is defined by simultaneous induction together with
their sets of roles $\roles(\GTIR{\GG}{\mathbf{I}})$ and
 {\em components}  $\components(\GTIR{\GG}{\mathbf{I}})$:
\begin{enumerate}[i)]
\item $\GTIR{G}{\mathbf{I}}\in \preGTIR$ and
$\roles(\GTIR{\GG}{\mathbf{I}}) = \roles(G)$ and 
$\components (\GTIR{G}{\mathbf{I}})  = \Set{G}$ if
	\begin{enumerate}
	\item $G$ is a global type of \gt,
	\item $\mathbf{I}\subseteq \roles(G)$,
\end{enumerate}
\item $\GTIR{\GTIR{{\GG}_1}{\mathbf{H}}\connect{\HH}{\KK}\GTIR{{\GG}_2}{\mathbf{K}}}{\mathbf{I}} \in \preGTIR$ 
	and
	$\roles(\GTIR{\GTIR{{\GG}_1}{\mathbf{H}}\connect{\HH}{\KK}\GTIR{{\GG}_2}{\mathbf{K}}}{\mathbf{I}}) = 
	\roles({\GTIR{{\GG}_1}{\mathbf{H}}}) \cup  \roles(\GTIR{{\GG}_2}{\mathbf{K}})$
and\\	
	$\components(\GTIR{\GTIR{{\GG}_1}{\mathbf{H}}\connect{\HH}{\KK}\GTIR{{\GG}_2}{\mathbf{K}}}{\mathbf{I}}) = 
	\components({\GTIR{{\GG}_1}{\mathbf{H}}}) \cup  \components(\GTIR{{\GG}_2}{\mathbf{K}})$ if
	\begin{enumerate}
	\item  $\GTIR{{\GG}_1}{\mathbf{H}}, \GTIR{{\GG}_2}{\mathbf{K}}\in\preGTIR$ with
		$\HH\in\mathbf{H}, \; \KK\in\mathbf{K}$,
	\item $\mathbf{I} = (\mathbf{H}\cup\mathbf{K})\setminus\Set{\HH,\KK}$,
	\item
	$\roles(\GTIR{{\GG}_1}{\mathbf{H}}) \cap \roles(\GTIR{{\GG}_2}{\mathbf{K}})
	=\emptyset.$
	\end{enumerate}
\end{enumerate}
\end{definition}

Notice that in a composed pre-GTIR, the notation of the set $\mathbf{I}$ is actually redundant and it is used just to immediately spot the interface roles.
By the above definition, a pre-GTIR is an expression formed by either a global type in \gt\ or a number of global types in \gt\ ``composed''  via symbols of the form $\ \connect{\HH}{\KK}$, and where sets of interface roles are identified by means of superscripts. 
These global types are what  we have defined as the {\em components} of the pre-GTIR.


\subsection{Working example}

We introduce the compatibility relation  we have in mind and the composition operator that we want to use
for constructing GTIRs
 by means of a working example inspired by one in \cite{deAlfaro2005}.
 
Let us assume we wish to develop an open system, let us dub it $S$, which can receive a text message from the outside. 
Once a text is received, the system tries to transmit it at most $n$ times (where also the number $n$ of possible trials is provided from the outside when the system is initialized).
A successful transmission is acknowledged by an {\sf ack} message; a {\sf nack} message
represents instead an unsuccessful transmission.
An {\sf ok} message is sent back in case of a successful transmission; a {\sf fail} message
in case of $n$ unsuccessful trials.
Before any transmission trial, a semantically-invariant transformation is applied to the message. 
The system can hence be used to send messages to social networks which are particularly
strict for what concerns propriety of language. If the message is not accepted by the
social network, our system automatically transforms it maintaining its sense, and sending it again and again up to $n$ times, invariantly transforming it each time. 
A counter is used to keep track of the number of  trials and it is reset to $n$ each time
a message is successfully transmitted. It is instead automatically reset to $n$ each time $0$ is reached, before 
issuing a failure message and restarting the protocol with some new message.\\


If we consider the formalism of generalized global types of \cite{DY12},
the overall behaviour of  the above system $S$ can be described by the
  graph $G$ shown in Fig.~\ref{fig:examplegg}, where the roles (participants) are:\\
{\tt M}: the manager of the system;\\
{\tt T}: the process implementing the semantically-invariant message transformation;\\
{\tt C}: the trials counter;\\
{\tt I}, {\tt J} and {\tt H}: the roles (that we identify as {\em interface roles}) representing those
parts of the environment which, respectively: initializes the system; sends the text message
and receives back the {\sf ok} or {\sf fail} message: receives the messages transmitted 
by the systems and acknowledges its propriety, if so.\\

 \begin{figure}[t]
\hrule

\vspace{4mm}
\begin{center}
      \begin{tikzpicture}
      
      \mkint{}{x9}[][H][ack][M];
      
      \mkint{below=0.5 of x9}{x11}[][M][reset][C];
       
      \mkint{below=0.5 of x11}{x13}[][M][ok][J];
       
       \mkint{right=of x11}{x15}[][C][zero][M];
       
       \mkint{right=of x15}{x16}[][C][notzero][M];
       
       \mkbranch{x15,x16}[x15x16][][-0.7][@][15pt];
       
        \mkint{above =0.18 of x15x16}{x10}[][H][nack][M];
        
        \mkint{below = 0.5 of x15}{x17}[][M][fail][J];
        
        \mkseq{x9}{x11};
        \mkseq{x11}{x13};
        \mkseq{x10}{x15x16};
        \mkseq{x15}{x17};
        
         \mkbranch{x9,x10}[x9x10][][0.8][@][15pt];
         
        \mkint{above = 0.5 of x9x10}{x6}[][M][text][H];
        
        \mkint{above = 0.5 of x6}{x7}[][T][text][M];

        \mkint{above = 0.5 of x7}{x5}[][M][text][T];
        
        \mkseq{x5}{x7};
        
         \mkseq{x7}{x6};
        
        \mkseq{x6}{x9x10};
        
        \mklooptwo[0.5][-97]{x5}{x16}{}{};

        \mkint{above = 1.3 of x5}{x2}[][J][text][M];
        
        \path[line] (x2) -- (entryx5);

        \mkmerge{x13,x17}[x13x17][][0.8][@][15pt];
        
        \mklooptwobelow[0.5][-24]{x2}{x13x17}{}{};
        
        \mkint{above = 1.3 of x2}{x0}[][I][trialsNum][C];
        
        \mkseq{x0}{entryx2};
        
        \node[source,above = 0.5 of x0] (src) {};
        
        \mkseq{src}{x0};
        
      \end{tikzpicture}
      \end{center}
\vspace{4mm}
\hrule
\caption{The global type $G$ of the working example}
\label{fig:examplegg}
\end{figure}
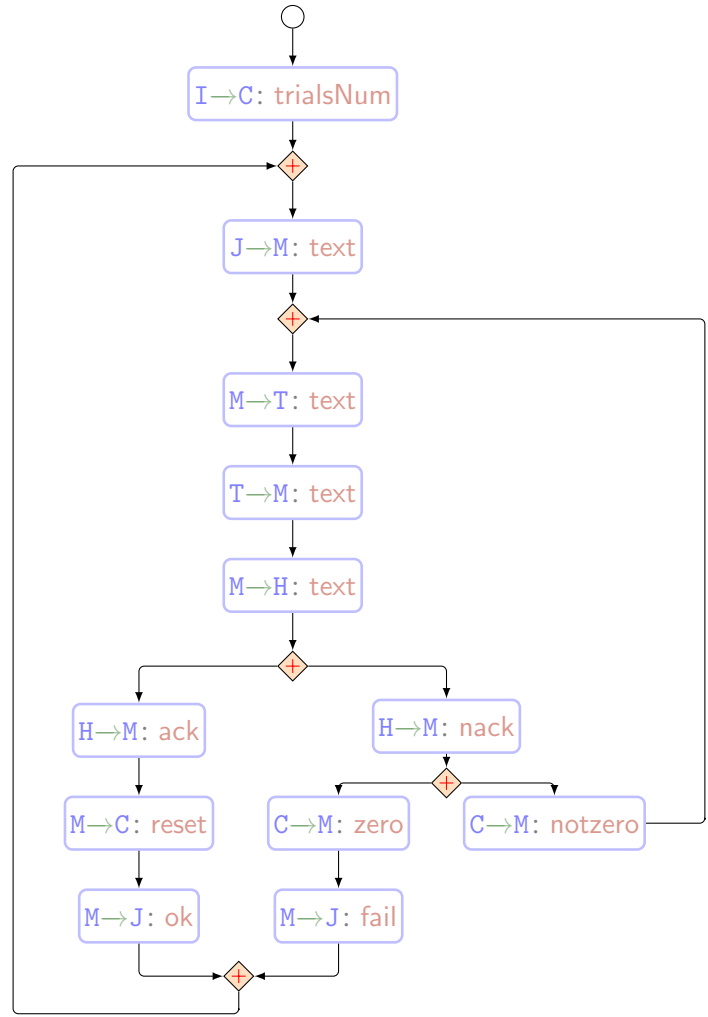

Informally, in the graph $G$ in Fig.~\ref{fig:examplegg}, a label $\tts\rightarrow \ttr : a$ represents an interaction where $\tts$ sends a message
$a$ to r. A vertex with label {\small $\bigcirc$} represents the source of the
graph and {\huge $\diamond$}\hspace{-8.5pt}{$^+$} marks vertexes corresponding
to branch or merge points, or to entry points of loops.
In our formalism, we can look at $G$ as a global type with interface roles by identifying the roles $\II,\JJ$ and $\HH$ as {\em interface roles}.
We do that by writing, according to the syntax of Def.~\ref{def:pre-gtir}(i):
                                                  $$\GTIR{G}{\Set{\II,\JJ,\HH}}$$
Given a global type with interface roles, it is reasonable to expect all its roles to be implemented but the interface ones,
since they are actually used to describe the behaviour of the ``environment'' of the system.
The projections of global types with interface roles onto their interface roles yield CFSMs which are used instead to check whether
two systems can be connected in a safe way. In particular, to check whether interface roles are ``compatible''.
 According to the projection algorithm for generalized global types (see \S 3.1 and Def. 3.4 of \cite{DY12}), the projection on role $\JJ$ of the graph $G$ of Fig.~\ref{fig:examplegg}  is the following  CFSM $M_\JJ = \projecton{G}{\JJ}$ (see (1))
which describes the behaviour of that part of the environment of system $S$ which sends a text and waits for a positive or negative answer.
\vspace{-2mm}
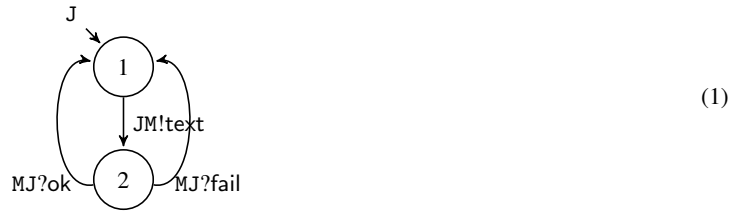
\begin{equation}
\footnotesize
\label{pic:I2}
      \begin{tikzpicture}[->,>=stealth',shorten >=1pt,auto,node distance=1.5cm,semithick]
  
  \node[state]   (one)                        {$1$};
  \node[draw=none,fill=none] (start) [above left = 0.3cm  of one]{{\tt J}};
  \node[state]            (two) [below of=one] {$2$};

   \path  (start) edge node {} (one) 
            (one)  edge                                   node {\edgelabel{JM}!{text}} (two)
           (two) edge       [bend left=100]       node [pos=0.1, left] {\edgelabel{MJ}{?}{ok}} (one)
                   edge        [bend right=100]      node [pos=0.1, right] {\edgelabel{MJ}{?}{fail}} (one);
           
       \end{tikzpicture}
\end{equation}

Going on with our example,
let us consider now another open system $S'$, having, among others, roles \AA, \BB\ and \KK, where \KK\ is one of its interface roles.
In $S'$, the roles {\tt A} and \BB\ keep on sending, in an alternating manner, a text message to \KK, which replies with a positive or negative
acknowledgement ({\sf ok} or {\sf fail}, respectively).
Role \BB\  can send its message only after a successful sending by \AA, and vice versa.
A {\sf fail} message from \KK\ forces the resending of the message.
Let us assume
that the behavior of interface role \KK\ is given by the following CFSM $M_\KK$ (see (2)):

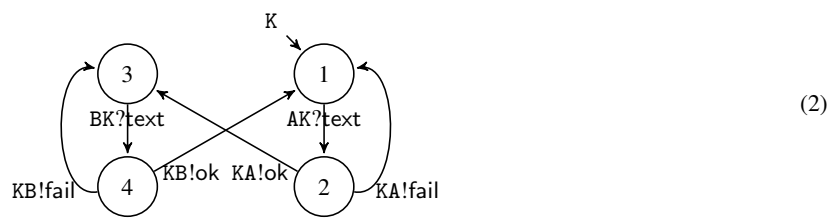
\begin{equation}
\footnotesize
\label{pic:J}
      \begin{tikzpicture}[->,>=stealth',shorten >=1pt,auto,node distance=1.5cm,semithick]
  
  \node[state]   (one)                        {$1$};
    \node[draw=none,fill=none] (start) [above left = 0.3cm  of one]{{\tt K}};
  \node[state]            (two) [below of=one] {$2$};
  \node[state]            (three) [left = 1.8cm of one] {$3$};
  \node[state]            (four) [left = 1.8cm of two] {$4$};

   \path (start) edge node {} (one) 
           (one)  edge                                    node [pos=0.6, above]    {\edgelabel{AK}?{text}} (two)
           (two) edge                                      node [pos=0.0, left]        {\edgelabel{KA}{!}{ok}} (three)
                   edge        [bend right=100]      node [pos=0.1, right]      {\edgelabel{KA}{!}{fail}} (one)
           (three)  edge                                   node [pos=0.6, above]    {\edgelabel{BK}?{text}} (four)
            (four) edge                                     node [pos=0.0, right]        {\edgelabel{KB}{!}{ok}} (one)
                      edge        [bend left=100]      node [pos=0.1, left]      {\edgelabel{KB}{!}{fail}} (three);
           
       \end{tikzpicture}
\end{equation}

The interface roles $\JJ$ and $\KK$ are {\em compatible}, in that
the {\sf text} message asked for by $\KK$ can be the one provided by $\JJ$ to system $S$, whereas the
{\sf ok} and {\sf fail} messages $\JJ$ receives can be the ones that $\KK$ sends to system $S'$.
In a nutshell, if we do not take into account channels in the labels, the language accepted by $\JJ$ is the 
dual (i.e. '!' and '?' are exchanged) of that accepted by $\KK$.

Once the compatibility of $\JJ$ and $\KK$ is ascertained, the behaviours of two {\em gateways} processes could be easily
constructed from $M_\JJ$ and $M_\KK$. 
The idea is to insert an intermediate state with appropriate transitions in the middle of any transition of $M_\JJ$ (and similarly of $M_\KK$) enabling to pass messages from the interface role of one system to the other. For instance, the transition of $M_\JJ$ from state $1$ to state $2$ labelled with $\edgelabel{JM}{!}{text}$ is split into two transitions (see Fig. \ref{eq:JK}  , left), where $\JJ$ first receives a text from $\KK$ and then sends it to $\mathtt{M}$.

Such {\em gateways} processes can be constructed by means of an algorithm that we dub
$\gateway{\cdot}$. It takes two arguments: the CFSM to be transformed and
the name of the interface role of the other system, and returns a ``gateway''  CFSM which 
enables systems to interact.
For what concerns our example, by applying $\gateway{\cdot}$ to $M_\JJ$ and $\KK$
and by applying $\gateway{\cdot}$ to $M_\KK$ and $\JJ$, we get the two CFSMs depicted in Figure \ref{eq:JK}.\\

By assuming $\GTIR{\GG'}{\mathbf{I}\cup\Set{\KK}}$ to be the GTIR denoting the open system $S'$ above,
the pre-GTIR  
                         $$\GTIR{ \GTIR{G}{\Set{\II,\JJ,\HH}}  \connect{\JJ}{\KK}  \GTIR{\GG'}{\mathbf{I}\cup\Set{\KK}}  }{\mathbf{I}\cup\Set{\II,\HH}}$$
 is actually a proper GTIR, since the interface roles ${\JJ}$ and ${\KK}$ are compatible.
 Its semantics is the system obtained by connecting $S$ and $S'$ by means of the gateways $\gateway{M_\JJ,\KK}$ and $\gateway{M_\KK,\JJ}$.
Notice that the GTIR now exposes the remaining interface roles $\mathbf{I}\cup\Set{\II,\HH}$.

\begin{remark}
We could choose to connect systems by implementing a single  ``two-sided'' gateway process, but this would imply to  
change all the names in the channels of $S$ and $S'$ from and to $\JJ$ and $\KK$. This is not feasible
if, as it is likely, $S$ and $S'$ have been separately implemented. Also there would be no straightforward generation of a single ``two-sided'' gateway process.
\end{remark}

\begin{remark}
One could wonder what would change if, instead of using gateways, one simply renamed
    the target of communications to interface nodes.
That is unfeasible (unless a rather strict relation of compatibility were used). 
It is enough to take into account our working example: if we tried to rename the target of 
communications between $M_{\mathtt{M}}$ and $M_\JJ$, machine $M_{\mathtt{M}}$
should be completely rewritten. In fact, whereas both in $S$ and in the new system with gateways
  $M_{\mathtt{M}}$ receives the text from $M_\JJ$, in a new system without gateways
  $M_{\mathtt{M}}$ would receive the text both from $M_{\mathtt{A}}$ and $M_{\mathtt{B}}$.
\end{remark}

 \begin{figure}[t]
\footnotesize
\hrule
\vspace{4mm}
$$
\begin{array}{c@{\hspace{2cm}}c}
      \begin{tikzpicture}[->,>=stealth',shorten >=1pt,auto,node distance=1.5cm,semithick]
 
  \node[state]           (one)                        {$1$};
   \node[draw=none,fill=none] (start) [above left = 0.3cm  of one]{{\tt J}};
  \node[state]            (two) [below of=one] {$\widehat 1$};
  \node[state]           (three) [below of=two] {$2$};
  \node[state]           (four)  [left of=two] {$\widehat 2'$};
  \node[state]           (five) [right of=two] {$\widehat 2''$};

   \path  (start) edge node {} (one) 
            (one)  edge                                   node {\edgelabel{KJ}?{text}} (two)
           (two) edge                                     node {\edgelabel{JM}{!}{text}} (three)
           (four) edge       [bend left]              node [pos=0.1, left] {\edgelabel{JK}{!}{ok}} (one)
           (five) edge        [bend right]            node [pos=0.1,right] {\edgelabel{JK}{!}{fail}} (one)
           (three) edge       [bend left]       node [pos=0.1, left] {\edgelabel{MJ}{?}{ok}} (four)
                     edge        [bend right]      node [pos=0.1, right] {\edgelabel{MJ}{?}{fail}} (five);
    
       \end{tikzpicture}
       &
       \begin{tikzpicture}[->,>=stealth',shorten >=1pt,auto,node distance=1.5cm,semithick]
  
  \node[state]   (one)                        {$1$};
   \node[draw=none,fill=none] (start) [above left = 0.3cm  of one]{{\tt K}};
  \node[state]            (two) [below of=one] {$\widehat{1}$};
   \node[state]            (three) [below of=two] {$2$};
    \node[state]            (four) [left of=two] {$\widehat 2'$};
     \node[state]            (five) [right of=two] {$\widehat 2''$};
  \node[state]            (six) [left = 3.5cm of one] {$3$};
  \node[state]            (seven) [left = 3.5cm of two] {$\widehat 3$};
      \node[state]            (eight) [below of=seven] {$4$};
  \node[state]            (nine) [right = 0.8cm of seven] {$\widehat 4''$};
  \node[state]            (ten) [left = 0.8cm of seven] {$\widehat 4'$};

   \path  (start) edge node {} (one) 
            (one)  edge                                    node [pos=0.6, above]    {\edgelabel{AK}?{text}} (two)
           (two) edge                                      node [pos=0.6, above]        {\edgelabel{KJ}{!}{text}} (three)
           (three)  edge        [bend right]      node [pos=0.1, right]      {\edgelabel{JK}{?}{fail}} (five)
                         edge      [bend left]        node [pos=0.1, left]      {\edgelabel{JK}{?}{ok}} (four)
           (five)  edge          [bend right]          node [pos=0.3, right]    {\edgelabel{KA}!{fail}} (one)
           (six)  edge                                    node [pos=0.6, above]    {\edgelabel{BK}?{text}} (seven)
           (seven) edge                                      node [pos=0.6, above]        {\edgelabel{KJ}{!}{text}} (eight)
           (eight)  edge        [bend left]      node [pos=0.1, left]      {\edgelabel{JK}{?}{fail}} (ten)
                         edge      [bend right]        node [pos=0.1, right]      {\edgelabel{JK}{?}{ok}} (nine)
           (ten)  edge          [bend left]          node [pos=0.3, left]    {\edgelabel{KB}!{fail}} (six)
           (nine)  edge          [bend left]          node [pos=0.1, left]    {\edgelabel{KB}!{ok}} (one)
           (four)  edge          [bend right]          node [pos=0.1, right]    {\edgelabel{KA}!{ok}} (six);
           
       \end{tikzpicture}
\end{array}
$$
\vspace{4mm}
\hrule
\caption{ $\gateway{M_\JJ,\KK}$ and $\gateway{M_\KK,\JJ}$ }
\label{eq:JK}
\end{figure}

\subsection{Interface compatibility and GTIRs}

The following definitions are preliminary to the formal definition of
interface compatibility and hence of a proper GTIR. 
First, we extend the projection function of the global type formalism \gt\ to pre-GTIRs.
The projection of interface roles to CFSMs can then be used below  to check interface compatibility.


\begin{definition}
Let $\GTIR{{\GG}}{\mathbf{I}}$ be a pre-GTIR and let ${\tt p}\in  \roles(\GTIR{{\GG}}{\mathbf{I}}) $. 
We define
$$\projecton{\GTIR{{\GG}}{\mathbf{I}}}{{\tt p}} =
                  \projecton{G}{{\tt p}} \hspace{3mm}\text{ where } G\in\components(\GTIR{{\GG}}{\mathbf{I}}) \text{ such that } {\tt p}\in \roles({G}).$$

\end{definition}
                                                                 

At next we want to consider the dual of the language accepted by a CFSM
when input and output are reversed and the names of communication
channels are forgotten. For that purpose we need the following definition:

\begin{definition}
\begin{enumerate}[i)]
\item
Let $\varphi\in\Act^*$, we define $\varphi^\mC\in(\Set{!,?}\times\mathbb{A})^*$ inductively by:\\
\centerline{
$\varepsilon^\mC = \varepsilon$ \hspace{8mm}$(\ttp\ttq ?a\cdot \varphi)^\mC= ?a\cdot \varphi^\mC$  \hspace{8mm} $(\ttp\ttq !a\cdot \varphi)^\mC= !a\cdot \varphi^\mC$.
}
Moreover, for $A\subseteq\Act^*$, $A^\mC = \Set{\varphi^\mC \mid \varphi\in A}$.
\item
We define $\Dual{(\cdot)} : (\Set{!,?}\times\mathbb{A}) \rightarrow (\Set{!,?}\times\mathbb{A})$ by:\hspace{4mm}
$\Dual{!a} = ? a$ \hspace{8mm} $\Dual{?a} = !a$.\\
$\Dual{(\cdot)}$ is then straightforwardly extended also to words and finite sets of words.
\end{enumerate}
\end{definition}                                                               
 
 Finally, we define interface compatibility by requiring that the CFSMs
 of two interface roles are dual to each other. 
 Additionally we require the absence of mixed states as well as input and output
determinism for each of the two CFSMs. In fact, if either of these two conditions were omitted
one can provide counterexamples showing that our results on preservation of safety properties
would generally no longer be valid.

                                                                 
\begin{definition}[Interface compatibility]\label{interface-comp}
\hfill

\begin{enumerate}[i)]
\item
\label{interface-comp-i}
Let $M$  and $M'$ be two CFSMs over $\roles$ and $\mathbb{A}$
($\roles'$ and $\mathbb{A}'$ resp.). $M$  and $M'$ are {\em compatible}, denoted by $M\interfacecomp M'$, whenever
\begin{enumerate}[1)]
\item
$\lang{M}^\mC = \Dual{\lang{M'}^\mC}.$
\item 
\label{interface-comp-ia}
$M$ and $M'$ do not contain mixed states.
\item 
$M$ and $M'$ are  ?!-deterministic.
\end{enumerate}
\item
\label{interface-comp-ib}
Let $\GTIR{\GG_1}{\mathbf{H}}$ and $\GTIR{\GG_2}{\mathbf{K}}$
be two pre-GTIRs. 
Two interface roles $\HH \in \mathbf{H}$ and $\KK \in \mathbf{K}$ are \emph{interface compatible},
denoted by ${\HH}\interfacecomp {\KK}$, if\,
$\projecton{\GTIR{\GG_1}{\mathbf{H}}}{{\HH}}\interfacecomp\projecton{\GTIR{\GG_2}{\mathbf{K}}}{{\KK}}$.

\end{enumerate}
\end{definition}

\noindent
It is easy to check that in our working example we have $\JJ \interfacecomp \HH$, since the CFSMs (\ref{pic:I2}) and (\ref{pic:J})  are compatible, i.e. one accepts the dual language of the other if channel names are not taken into account;
they have no mixed states and are  ?!-deterministic.

We are now ready to introduce our notion of GTIR.
The syntactic construction of proper GTIRs follows the construction of pre-GTIRs but imposes two semantic conditions for the underlying CFSMs formulated in  i) and ii) of the next definition.


 \begin{definition}[GTIR]\hfill
\label{def:gtir}
\begin{enumerate}[i)]
\item
A pre-GTIR $\GTIR{G}{\mathbf{I}}$ formed by a global type $G$ and interface roles
$\mathbf{I}\subseteq \roles(G)$ is a GTIR if
no communication between interface roles (i.e. roles in $\mathbf{I}$) is present, 
i.e.\ for each $p \in \roles(G)$ the projection
$\projecton{G}{\ttp}$ has no transition with a label of the form
$\II\JJ!a$ or $\II\JJ?a$ with $\II,\JJ\in\mathbf{I}$.
\item
A pre-GTIR $\GTIR{\GTIR{{\GG}_1}{\mathbf{H}}\connect{\HH}{\KK}\GTIR{{\GG}_2}{\mathbf{K}}}{\mathbf{I}}$
obtained by the composition of two GTIRs $\GTIR{\GG_1}{\mathbf{H}}$
and $\GTIR{\GG_2}{\mathbf{K}}$ via interface roles $\HH\in\mathbf{H}, \; \KK\in\mathbf{K}$ is a GTIR if\,  $\HH$ and $\KK$ are interface compatible, i.e.\ ${\HH}\interfacecomp {\KK}$.
\end{enumerate}
\end{definition}                                                        


%

\subsection{Semantics of GTIRs}

To provide semantics for GTIRs we first define the {\em gateway} transformation $\gateway{\cdot}$ previously mentioned.
By means of such a function it is possible to construct the gateway processes enabling the CFSM systems described by
GTIRs to be connected.

The $\gateway{\cdot}$ 
function takes as input a CFSM $M_\HH$ of some role $\HH$ and a role name $\KK$. In our application $\HH$ and $\KK$ will be interface roles where $\HH$ is the interface role to be connected to the interface role  $\KK$. The gateway function transforms $M_\HH$ by ``inserting''  a new state ``in between'' any 
transition. In such a way a transition from $q$ to $q'$ receiving a message $a$ from a role $\tts (\neq \KK)$ is transformed
into two transitions: one from $q$ to the new state $\widehat{q}$ receiving $a$ from $\tts$, and one from
$\widehat{q}$ to $q'$ sending $a$ to $\KK$. 
Conversely, a transition from $q$ to $q'$ sending a message $a$ to a role $\tts  (\neq \KK)$ is transformed
into two transitions: one from $q$ to the new state $\widehat{q}$ receiving $a$ from $\KK$, and one from
$\widehat{q}$ to $q'$ sending $a$ to $\tts$. We distinguish the new ``inserted'' states by superscripting them
by the transition they are ``inserted in between''.

\begin{definition}[The $\gateway{\cdot}$ transformation]\hfill\\
\label{def:gateway}
Let $M_\HH=(Q,q_0,\mathbb{A},\delta)$ be the CFSM of a role $\HH$, and let  $\KK$ be a role name.
We define 
$$\gateway{M_\HH,\KK}=(Q',q_0,\mathbb{A},\delta')$$
where\\
-  
$Q' = Q \cup \widehat{Q_\delta}$, with 
$\widehat{Q_\delta} =\bigcup_{q\in Q}\Set{q^{(q,\elle,q')} \mid (q,\elle,q')\in\delta}$, and\\
- 
$\delta' = \hspace{12pt}\Set{(q,\KK\HH?a,q^{(q,\HH\tts!a,q')}), (q^{(q,\HH\tts!a,q')},\HH\tts!a,q') \mid  (q,\HH\tts!a,q')\in\delta}$\\
\hphantom{$\bullet a \delta' =$}$\cup\  
\Set{(q,\tts\HH?a,q^{(q,\tts\HH?a,q')}), (q^{(q,\tts\HH?a,q')},\HH\KK!a,q') \mid  (q,\tts\HH?a,q')\in\delta}$
\end{definition}
\medskip
\noindent
For the sake of readability, we shall often denote elements $q^{(q,\elle,q')}, q^{(q,\elle',q'')},q^{(q,\elle'',q''')}, \ldots$
of  $\widehat{Q_\delta}$ \linebreak
by $\widehat{q}, \widehat{q'}, \widehat{q''},\ldots$.
We shall also refer to $\widehat{Q_\delta}$ simply as $\widehat{Q}$ when clear from the context.

We can now define the composition of two communicating systems $S_1$ and $S_2$ w.r.t.\ compatible interface roles $\HH$ and $\KK$.
We take the union of the CFSMs of $S_1$ and $S_2$ \emph{but} replace the CFSMs
$M_\HH$ and $M_\KK$ of the interface roles
$\HH$ and $\KK$ by their \emph{gateway CFSMs} $\gateway{M_\HH,\KK}$ and $\gateway{M_\KK,\HH}$.

\begin{definition}[Composition of communicating systems]\hfill\\
\label{def.semconn}
Let $S_1 = (M^1_\ttp)_{\ttp\in\roles_1}$ and $S_2 = (M^2_\ttq)_{\ttq\in\roles_2}$ be two communicating systems
over $\roles_1$ and $\mathbb{A}_1$ ($\roles_2$ and $\mathbb{A}_2$ resp.)
such that $\roles_1\cap\roles_2 = \emptyset$.
 Moreover, let $\HH\in\roles_1$ and $\KK\in\roles_2$  be such that
 $\HH \interfacecomp \KK$ (i.e.\ $M^1_\HH \interfacecomp M^2_\KK$). \\ 
The \emph{composition of $S_1$ and $S_2$ w.r.t. $\HH$ and $\KK$} is the communicating system
                        $${S_{1}}\connect{\HH}{\KK} {S_{2}} = (M_\ttp)_{\ttp\in(\roles_1\cup\roles_2)}$$
over $\roles_1\cup\roles_2$ and $\mathbb{A}_1 \cup \mathbb{A}_2$
where $M_\HH = \gateway{M^1_\HH, \KK}$, $M_\KK = \gateway{M^2_\KK, \HH}$, 
$M_\ttp =M^1_\ttp$ for all $\ttp\in\roles_1$ and
$M_\ttp =M^2_\ttp$ for all $\ttp\in\roles_2$.\footnote{The CFSMs over $\roles_1$ and $\mathbb{A}_1$
 ($\roles_2$ and $\mathbb{A}_2$ resp.) are considered here as CFSMs over $\roles_1\cup\roles_2$
 and $\mathbb{A}_1 \cup \mathbb{A}_2$.}
\end{definition}

The semantics of a GTIR is inductively defined following its syntactic construction. 
In the base case, its semantics is the communicating system obtained by the projections
of the underlying global graph. The semantics of a composite GTIR is the composition
of the CFSM systems denoted by its constituent parts.

\begin{definition}[GTIR semantics]\hfill\\
\label{def.semgtir}
The communicating system $\Sem{\GTIR{{\GG}}{\mathbf{I}}}$ denoted by a GTIR $\GTIR{{\GG}}{\mathbf{I}}$ 
is inductively defined as follows:
\begin{itemize}
\item[-]
$\Sem{\GTIR{{G}}{\mathbf{I}}} =  (\projecton{G}{\ttp})_{\ttp\in\roles(G)}$   where $G$ is a global type in \gt;

\item[-]
$\Sem{\GTIR{\GTIR{{\GG}_1}{\mathbf{H}}\connect{\HH}{\KK}\GTIR{{\GG}_2}{\mathbf{K}}}{\mathbf{I}}}
= \Sem{   \GTIR{{\GG}_1}{\mathbf{H}}  } \connect{\HH}{\KK} \Sem{\GTIR{{\GG}_2}{\mathbf{K}}}.$
\end{itemize}
\end{definition}
\noindent

It is immediate to check that the operation of  ``connecting'' GTIRs is semantically commutative and associative,
i.e. the following holds:\\
({\em comm}) \hspace{4pt}
$ \Sem{\GTIR{\GTIR{{\GG}_1}{\mathbf{H}}\connect{\HH}{\KK}\GTIR{{\GG}_2}{\mathbf{K}}}{\mathbf{I}}}
 = \Sem{\GTIR{\GTIR{{\GG}_2}{\mathbf{K}}\connect{\KK}{\HH}\GTIR{{\GG}_1}{\mathbf{H}}}{\mathbf{I}}} $ \\
({\em ass})\hspace{19pt}
 $\Sem{
  \GTIR{\GTIR{\GTIR{{\GG}_1}{\mathbf{H}}\connect{\HH}{\KK}\GTIR{{\GG}_2}{\mathbf{K}}}{\mathbf{I}}
           \connect{\II}{\JJ}
            \GTIR{{\GG}_3}{\mathbf{J}} }
           {\mathbf{I'}}
 }
 =
 \Sem{\GTIR{
                    \GTIR{{\GG}_1}{\mathbf{H}}\connect{\HH}{\KK} 
                                     \GTIR{
                                               \GTIR{{\GG}_2}{\mathbf{K}}  \connect{\II}{\JJ} \GTIR{{\GG}_3}{\mathbf{J}}
                                               }{\mathbf{J'}}
                               }{\mathbf{I'}}
 }  
 $\\
 \hphantom{({\em ass})\hspace{19pt}} where {\scriptsize ${\mathbf{J'}} = \mathbf{K}\cup\mathbf{J}\setminus\Set{\II,\JJ}$}\\

%
%
%


\section{Preservation of Safety-Properties}
\label{sect:safetypreservation}

In the present section we show that if we take two safe communicating systems
$S_1$ and $S_2$ such that $S_1$ possesses a CFSM $M^1_\HH$ and $S_2$ a CFSM $M^2_\KK$ which is compatible with $M^1_\HH$, replace both CFSMs by their gateway transformations and then join the resulting systems, we get a   a safe system.

\vspace{2mm}
\textbf{General assumption:} In the following of this section we generally assume given a
system  \linebreak 
{$S={S_{1}}\connect{\HH}{\KK} {S_{2}}$}  composed
as described in Def.~\ref{def.semconn}
from systems $S_1$ and $S_2$ with compatible CFSMs $M^1_\HH$ and $M^2_\KK$. 

\vspace{2mm}
\textbf{Notation:} 
The channels of $S$ are $C=\Set{\ttp\ttq \mid \ttp,\ttq\in \roles, \ttp\neq\ttq}$
and the channels of $S_i$ are \linebreak $C_i=\Set{\ttp\ttq \mid \ttp,\ttq\in \roles_i, \ttp\neq\ttq}$ for $i=1,2.$
If $s= (\vec{q},\vec{w})$ is a configuration of $S$, where $\vec{q}=(q_\ttp)_{\ttp\in\roles}$
and $\vec{w} = (w_{\ttp\ttq})_{\ttp\ttq\in C}$,
we write $\restrict{s}{i}$ for $(\restrict{\vec{q}}{i},\restrict{\vec{w}}{i})$ 
where $\restrict{\vec{q}}{i} = (q_\ttp)_{\ttp\in\roles_i}$ and 
$\restrict{\vec{w}}{i} =  (w_{\ttp\ttq})_{\ttp\ttq\in C_i}$ ($i = 1,2$).
Notice that $\restrict{s}{i}$ is not necessarily a configuration of $S_i$, because of possible states in $\widehat{Q}$,
which are the additional states of the gateways.\\

%

The following technical properties easily descend from the definition of $\gateway{\cdot}$.
In particular from the fact that the gateway transformation of a machine $M$ does insert an intermediate state
 between any pair of states of $M$ connected by a transition. By definition, the intermediate state
 possesses exactly one incoming transition and one outgoing transition. 

\begin{fact}
\label{fact:uniquesending}
Let $s= (\vec{q},\vec{w}) \in RS(S)$ be a reachable configuration of
$S ={S_{1}}\connect{\HH}{\KK} {S_{2}}$.
\begin{enumerate}
\item
\label{fact:uniquesending-i}
If ${q_\HH}= \widehat{q} \in\widehat{Q_\HH}$ then
${q_\HH}$ is not final and
 there exists a unique transition $({q_\HH},\_,\_)\in\delta_\HH$.\\
  Moreover such a transition is of the form
 $(q_\HH,\HH\tts!a,q')$ with $q'\not\in\widehat{Q_\HH}$.\\
Similarly for $\KK$.

\item
\label{fact:uniquesending-ii}
If ${q_\HH}\not\in\widehat{Q_\HH}$ then either $q_\HH$ is final, or any transition $({q_\HH},\_,\_)\in\delta_\HH$
is an input  one, that 
is of the form $({q_\HH},\tts\HH?a,\widehat{q'_\HH})$ with $\widehat{q'_\HH}\in\widehat{Q_\HH}$. Similarly for $\KK$.
\item
\label{fact:uniquesending-iii}
If ${q_\HH}\not\in\widehat{Q_\HH}$ then
             \begin{enumerate}[a)]
\item
If $(q_\HH,\KK\HH?a,\widehat{q'_\HH})\in\delta_\HH$  then there exists $(\widehat{q'_\HH},\HH\tts!a,q''_\HH)\in\delta_\HH$ with $\tts \neq \KK$ 
such that $(q_\HH,\HH\tts!a,q''_\HH)\in\delta^1_\HH$.
The same holds for $\delta^2_\KK$ by exchanging $\HH$ with $\KK$ and vice versa.
\item
If $(q_\HH,\tts\HH?a,\widehat{q'_\HH})\in\delta_\HH$ with $\tts\neq\KK$  then there exists   $(\widehat{q'_\HH},\HH\KK!a,q''_\HH)\in\delta_\HH$  
such that $(q_\HH,\tts\HH?a,q''_\HH)\in\delta^1_\HH$.
The same holds for $\delta^2_\KK$ by exchanging $\HH$ with $\KK$ and vice versa.
              \end{enumerate}
\end{enumerate}
\end{fact}

If a reachable configuration of the connected system $S={S_{1}}\connect{\HH}{\KK} {S_{2}}$ does not involve an intermediate state of the gateway
$M_\HH = \gateway{M^1_\HH, \KK}$, 
then by taking into account only the states of machines of $S_1$ and disregarding
the channels between the gateways, 
we get a
reachable configuration of $S_1$. Similarly for $S_2$.

\begin{lemma}\hfill\\
\label{lem:nohatrestrict}
Let $s= (\vec{q},\vec{w}) \in RS(S)$ be a reachable configuration of $S ={S_{1}}\connect{\HH}{\KK} {S_{2}}$.
\begin{enumerate}[i)]
\item
\label{lem:nohatrestrict-i}
$q_\HH\not\in\widehat{Q_\HH} \implies 
\restrict{s}{1}\in RS(S_1)$;
\item
\label{lem:nohatrestrict-ii}
$q_\KK\not\in\widehat{Q_\KK} \implies 
\restrict{s}{2}\in RS(S_2)$.
\end{enumerate}
\end{lemma}

\begin{proof}\,
[{\it Sketch}]
(\ref{lem:nohatrestrict-i})
If $s \in RS(S)$, then there exists $s_0\lts{}s_1\lts{} \ldots\lts{} s_{n-1}\lts{}s_n=s$. Let $s_i = (\vec{q_i},\vec{w_i})$ $(i=0,..n)$.
Let $j \geq 0$ be the smallest index such that ${q_j}_\HH\not\in \widehat{Q_\HH}$ and  ${q_{j+1}}_\HH\in \widehat{Q_\HH}$
(if there is not such a $j$, then the thesis follows immediately).
By definition of $\gateway{\cdot}$ we have that $s_j\lts{\ttr\HH?a} s_{j+1}$ for a certain $\ttr$.
Now let $t\geq j+1$ be the smallest index such that ${q_t}_\HH = {q_{j+1}}_\HH$ and ${q_{t+1}}_\HH\not\in \widehat{Q_\HH}$.
Such an index $t$ does exist because of the hypothesis $q_\HH\not\in\widehat{Q_\HH}$
(moreover, notice that no-self loop transitions are possible out of a state in $\widehat{Q_\HH}$).
By definition of $\gateway{\cdot}$ we have that $s_{t}\lts{\HH\tts!a} s_{t+1}$ for a certain $\tts$.
It is now possible to build a configuration-transitions sequence like the following one\\
\centerline{
$s_0\lts{}s_1\lts{}\ldots s_j\lts{\ttr\HH?a} s_{j+1}\lts{\HH\tts!a} s'_{j+2} \lts{} \ldots \lts{}s'_{n-1}\lts{}s_n=s$
}
where ${q_{j+2}}_\HH\not\in \widehat{Q_\HH}$.\\
By iterating this procedure, we can get a sequence
\\
\centerline{
$s_0\lts{}s_1\lts{}\ldots s_j\lts{} s_{j+1}\lts{} s'_{j+2}\lts{} s''_{j+3} \lts{} \ldots \lts{}s''_{n-1}\lts{}s_n=s$
}
such that any transition of the form $\ttr\HH?a$ is immediately followed by a transition $\HH\tts!a$.

Now, it is possible to check that
\begin{enumerate}[a)]
\item
 for $z=0..j-1$, either $\restrict{s_z}{1} \ltsone{} \restrict{s_{z+1}}{1}$
or $\restrict{s_z}{1} = \restrict{s_{z+1}}{1}$, and

\item
$\restrict{s_j}{1} \ltsone{} \restrict{s'_{j+2}}{1}$.
\end{enumerate}

By doing that for any transition of the form $\ttr\HH?a$ which is immediately followed by a transition $\HH\tts!a$,
we can get a sequence $\restrict{s_0}{1} \lts{}^* \restrict{s}{1}$. So $\restrict{s}{1}\in RS(S_1)$.\\
(\ref{lem:nohatrestrict-ii}) This case can be treated similarly to (\ref{lem:nohatrestrict-i}).
\end{proof}


In a reachable configuration of a connected system, if the states of the gateways are not 
among those introduced by the transformation $\gateway{\cdot}$ and not final, and if the channels between
the gateways are empty, then one of the two gateways is ready to receive messages that
the other gateway is ready to receive from its system's participants.
This property relies on compatibility.

\begin{lemma}
\label{lem:alltrans}
Let $s= (\vec{q},\vec{\varepsilon}) \in RS(S)$ be a reachable configuration of
$S={S_{1}}\connect{\HH}{\KK} {S_{2}}$
such that 
\begin{enumerate}
\item 
 $q_\HH\not\in\widehat{Q_\HH}$ and $q_\KK\not\in\widehat{Q_\KK}$, and
 \item 
 $q_\HH$ and $q_\KK$ are not final.
 \end{enumerate}
 Then either
 \begin{enumerate}[a)]
 \item
 \label{lem:alltrans-a}
all the transitions from $q_\HH$ in $\delta_\HH$ are of the form $(q_\HH,\KK\HH?\_,\_)$ and\\
all the transitions from $q_\KK$ in $\delta_\KK$ are of the form $(q_\KK,\tts\KK?\_,\_)$ with $\tts\neq\HH$, or
\item
\label{lem:alltrans-b}
all the transitions from $q_\KK$ in $\delta_\KK$ are of the form $(q_\KK,\HH\KK?\_,\_)$ and\\
all the transitions from $q_\HH$ in $\delta_\HH$ are of the form $(q_\HH,\tts\HH?\_,\_)$ with $\tts\neq\KK$
\end{enumerate}
\end{lemma}

\begin{proof}\,
[{\it Sketch}]
Let \\
\centerline{
$s_0\lts{\elle_1}s_1\lts{\elle_2} \ldots \lts{\elle_{n-1}}s_{n-1}\lts{\elle_n}s_n=s$
}
be a configuration-transitions sequence leading to $s\in RS(S)$.
Let $\restrictup{s}{\HH}$ be the sequence of transitions  $\lts{\elle_{i_{1}}}\ldots\lts{\elle_{i_{\HH}}}$
of the above sequence such that, for any $m\in\Set{i_1,..i_\HH}$, $\elle_{m}\in\delta_\HH$.
We define similarly the sequence $\restrictup{s}{\KK}$.
By definition of $\gateway{\cdot}$, and by the fact that
 $q_\HH\not\in\widehat{Q_\HH}$,   
 we have that $\restrictup{s}{\HH}$ is made of consecutive pairs of the form
$\lts{(\_,\KK\HH?a,\_)}\lts{(\_,\HH\tts!a,\_)}$, with $\tts\neq\KK$, or $\lts{(\_,\tts\HH?a,\_)}\lts{(\_,\HH\KK!a,\_)}$, with $\tts\neq\KK$.
Similarly for $\restrictup{s}{\KK}$.\\
Since $w_{\HH\KK} =w_{\KK\HH} = \varepsilon$ (which immediately follows from the hypothesis
$s= (\vec{q},\vec{\varepsilon})$), the number of 
pairs $\lts{(\_,\KK\HH?\_,\_)}\lts{(\_,\HH\tts!\_,\_)}$ in $\restrictup{s}{\HH}$ is equal to the number of pairs
$\lts{(\_,\ttr\KK?\_,\_)}\lts{(\_,\KK\HH!\_,\_)}$ in $\restrictup{s}{\KK}$; and vice versa.
This implies that $\mid\restrictup{s}{\HH}\mid = \mid  \restrictup{s}{\KK}\mid$.\\
Now, by extending to sequences $\restrictup{s}{\KK}$ and $\restrictup{s}{\HH}$
the following symbols function on pairs\\ 
\centerline{
$\mathit{symb}(\lts{(\_,\KK\HH?a,\_)}\lts{(\_,\HH\tts!a,\_)}) =!a$ \hspace{5mm} $\mathit{symb}(\lts{(\_,\tts\HH?a,\_)}\lts{(\_,\HH\KK!a,\_)}) =?a$
}
(and the clauses got by exchanging $\HH$ and $\KK$),
we get $\mathit{symb}(\restrictup{s}{\HH}) \in \lang{M^1_\HH}^\mC$ and $\mathit{symb}(\restrictup{s}{\KK}) \in \lang{M^2_\KK}^\mC$. Moreover, 
$\mathit{symb}(\restrictup{s}{\HH}) = \Dual{\mathit{symb}(\restrictup{s}{\KK})}$.
Notice that, by  ?!-determinism of $M^1_\HH$ and $M^2_\KK$ and absence of mixed states, there are no other sequences of pairs
for which the previous properties hold.\\
Now, by contradiction, and by recalling that $q_\HH$ and $q_\KK$ are not final,  let us assume that either
\begin{enumerate}[I)]
\item
\label{I}
all the transitions from $q_\HH$ in $\delta_\HH$ are of the form $(q_\HH,\KK\HH?\_,\_)$ and\\
 all the transitions from $q_\KK$ in $\delta_\KK$ are of the form$(q_\KK,\HH\KK?\_,\_)$  or
\item
\label{II}
all the transitions from $q_\HH$ in $\delta_\HH$ are of the form $(q_\HH,\tts\HH?\_,\_)$ with $\tts\neq\HH$ and\\
 all the transitions from $q_\KK$ in $\delta_\KK$ are of the form $(q_\KK,\tts\KK?\_,\_)$ with $\tts\neq\HH$.
\end{enumerate}
Notice that no other possibilities are given because, by compatibility, $M^1_\HH$ and $M^2_\KK$ have no mixed state.\\
If  (\ref{I}) holds, we get a contradictrion, since, by Fact \ref{fact:uniquesending}(\ref{fact:uniquesending-iii}),
we would get both $\mathit{symb}(\restrictup{s}{\HH})\cdot !a\in \lang{M^1_\HH}^\mC$ and
$\mathit{symb}(\restrictup{s}{\KK})\cdot !b \in \lang{M^2_\KK}^\mC$ for some $a$ and $b$, which is impossible by compatibility.\\
In case (\ref{II}) we get a contradiction by arguing analogously as in the previous case.
\end{proof}



\begin{lemma}\label{lem:restrRS}
Let $s= (\vec{q},\vec{\varepsilon}) \in RS(S)$ be a deadlock configuration for $S$.
Then, either $\restrict{s}{1}\in RS(S_1)$ is a deadlock configuration for $S_1$
or $\restrict{s}{2}\in RS(S_2)$ is a deadlock configuration for $S_2$ (or both).
\end{lemma}

\begin{proof}
By definition of deadlock configuration and by Fact \ref{fact:uniquesending}(\ref{fact:uniquesending-i}), we have that
neither $q_\HH\in\widehat{Q_\HH}$ nor $q_\KK\in\widehat{Q_\KK}$. Otherwise
there will be an output transition from either $q_\HH$ or  $q_\KK$, contradicting $s$ to be a deadlock configuration.
Hence necessarily $q_\HH\not\in\widehat{Q_\HH}$ and $q_\KK\not\in\widehat{Q_\KK}$. 
So, by Lemma \ref{lem:nohatrestrict} we get $\restrict{s}{i}\in RS(S_i)$ for $i=1,2$.
We show that under the assumption either $\restrict{s}{1}\in RS(S_1)$ or $\restrict{s}{2}\in RS(S_2)$ is a deadlock configuration.

Since $s= (\vec{q},\vec{\varepsilon}) \in RS(S)$ is a deadlock configuration for $S$, we have
that for all $\ttr\in\roles$, $q_r$ is a receiving state of $M_\ttr$.
In particular, for all $\ttr\in \roles_1\setminus\Set{\HH}$, $q_r$ is a receiving state
of $M_\ttr = M^1_\ttr$ and
for all $\ttr\in \roles_2\setminus\Set{\KK}$, $q_r$ is a receiving state of $M_\ttr = M^2_\ttr$. 
It remains to show that either $q_\HH$ is a receiving state of $M^1_\HH$
or $q_\KK$ is a receiving state of $M^1_\KK$,
whereby we can assume that $q_\HH$ is a receiving state of
$M_\HH = \gateway{M^1_\HH, \KK}$
and $q_\KK$ is a receiving state of $M_\KK = \gateway{M^2_\KK, \HH}$.

Without loss of generality we consider $q_\HH$. The proof for $q_\KK$ is analogous.
By Lemma~\ref{lem:alltrans} we have two possibilities to take into account
         \begin{description}
         \item
         All the transitions from $q_\HH$ in $\delta_\HH$ are of the form $(q_\HH,\KK\HH?\_,\_)$.\\
         In such a case, still resorting to Lemma \ref{lem:alltrans}, we have that 
         all the transitions from $q_\KK$ in $\delta_\KK$ are of the form $(q_\KK,\tts\KK?\_,\_)$ with $\tts\neq\HH$
         and, by Fact  \ref{fact:uniquesending}(\ref{fact:uniquesending-iii}),\\
          $(q_\KK,\tts\KK?a,\widehat{q'_\KK})\in\delta_\KK$  implies 
          $(\widehat{q'_\KK},\KK\HH!a,q''_\KK)\in\delta_\KK$ and
          $(q_\KK,\tts\KK?a, q''_\KK)\in\delta^2_\KK$.\\
          Hence $q_\KK$ is a receiving state
          of $M^{2}_\KK$. In summary,
           $\restrict{s}{2}$ is a deadlock configuration for $S_2$.
          \item
          All the transitions from $q_\HH$ in $\delta_\HH$ are of the form $(q_\HH,\tts\HH?\_,\_)$ with $\tts\neq\HH$.\\
          In such a case, by Fact \ref{fact:uniquesending}(\ref{fact:uniquesending-iii}), we have that \\
          $(q_\HH,\tts\HH?a,\widehat{q'_\HH})\in\delta_\HH$  implies 
          $(\widehat{q'_\HH},\HH\KK!a,q''_\HH)\in\delta_\HH$ and
          $(q_\HH,\tts\HH?a,q''_\HH)\in\delta^1_\HH.$\\
           Hence $q_\HH$ is a receiving state
          of $M^{1}_\HH$. In summary, $\restrict{s}{1}$ is a deadlock configuration for $S_1$.
          \end{description}  
\end{proof}

\begin{corollary}[Preservation of deadlock-freeness]\hfill\\
\label{prop:dfPreservation}
Let $S_1$ and $S_2$ be deadlock-free.
Then $S = {S_{1}} \connect{\HH}{\KK} {S_{2}}$ is  deadlock-free.
\end{corollary}
\begin{proof}
By contradiction, let us assume there is an $s\in RS(S)$ such that $s=(\vec{q},\vec{\varepsilon})$ is a deadlock configuration.
We get immediately a contradiction by Lemma \ref{lem:restrRS} and the fact that $S_1$ and $S_2$ are two deadlock-free systems.
\end{proof}

Compatibility of interface roles forces all the messages sent by a gateway to be correctly received by the other one.
This implies that if the gateways both reach final states, the channels connecting them are empty.

\begin{lemma}
\label{lem:wempty}
If $s= (\vec{q},\vec{w}) \in RS(S)$ is a reachable configuration of $S={S_{1}}\connect{\HH}{\KK} {S_{2}}$ 
such that both  states $q_\HH$ and $q_\KK$ are final,
then $w_{\HH\KK} = w_{\KK\HH} = \varepsilon$.
\end{lemma}

\begin{proof}\,
[{\it Sketch}]
By Fact. \ref{fact:uniquesending}(\ref{fact:uniquesending-i}), $q_\HH\notin \widehat{Q_\HH}$ and $q_\KK\notin \widehat{Q_\KK}$. We now proceed as in the first part of the proof of Lemma \ref{lem:alltrans}.
Let \\
\centerline{
$s_0\lts{\elle_1}s_1\lts{\elle_2} \ldots \lts{\elle_{n-1}}s_{n-1}\lts{\elle_n}s_n=s$
}
be a configuration-transitions sequence leading to $s\in RS(S)$.\\
Let $\restrictup{s}{\HH}$ be the sequence of transitions  $\lts{\elle_{i_1}}\ldots\lts{\elle_{i_\HH}}$
of the above sequence such that, for any $m\in\Set{i_1,..i_\HH}$, $\elle_{m}\in\delta_\HH$.
We define similarly the sequence $\restrictup{s}{\KK}$.
By definition of $\gateway{\cdot}$, and by the fact that
 $q_\HH\not\in\widehat{Q_\HH}$,   
 we have that $\restrictup{s}{\HH}$ is made of consecutive pairs of the form
$\lts{(\_,\KK\HH?a,\_)}\lts{(\_,\HH\tts!a,\_)}$, with $\tts\neq\KK$, or
$\lts{(\_,\tts\HH?a,\_)}\lts{(\_,\HH\KK!a,\_)}$, with $\tts\neq\KK$.
Similarly for $\restrictup{s}{\KK}$.\\
Without loss of generality, let us assume $\restrictup{s}{\HH}$ to begin with a pair of the form
$\lts{(\_,\tts\HH?a,\_)}\lts{(\_,\HH\KK!a,\_)}$.
(Otherwise $\restrictup{s}{\KK}$ would begin with a pair of the form
$\lts{(\_,\tts\KK?a,\_)}\lts{(\_,\KK\HH!a,\_)}$
since $\lang{M^1_\HH}^\mC = \Dual{\lang{M^2_\KK}^\mC}$).)
Hence, up to the role $\tts$, the symbols of $\restrictup{s}{\KK}$ are uniquely determined by $\restrictup{s}{\HH}$ because of
?!-determinism and absence of mixed states.
Then $\mathit{symb}(\restrictup{s}{\HH}) \in \lang{M^1_\HH}$ and $\mathit{symb}(\restrictup{s}{\KK}) \in \lang{M^2_\KK}$. Moreover, 
$\mathit{symb}(\restrictup{s}{\HH}) = \Dual{\mathit{symb}(\restrictup{s}{\KK})}$.
By assuming either $w_{\HH\KK} \neq \varepsilon$ or $w_{\KK\HH} \neq \varepsilon$ we would get a contradiction.
In fact, by the above, $\mid\restrictup{s}{\HH}\mid = \mid  \restrictup{s}{\KK}\mid$.
\end{proof}

\begin{lemma}
\label{lem:restrRSom}
Let $s= (\vec{q},\vec{w}) \in RS(S)$ be an orphan-message configuration for $S$.
Then,
either $\restrict{s}{1}$ is an  orphan-message configuration for $S_1$
or $\restrict{s}{2}$ is an  orphan-message configuration for $S_2$.
\end{lemma}

\begin{proof}
By Fact \ref{fact:uniquesending}(\ref{fact:uniquesending-i}), no state in $\widehat{Q_\HH}\cup \widehat{Q_\KK}$ can be final. So, by definition
of orphan-message configuration, $q_\HH\not\in \widehat{Q_\HH}$ and $q_\KK\not\in \widehat{Q_\KK}$.
Hence, for $i=1,2$, $\restrict{s}{i}\in RS(S_i)$ by Lemma \ref{lem:nohatrestrict}. 
Now, by definition of orphan-message configuration, both  $q_\HH$ and $q_\KK$ are final
in $M_\HH$ and $M_\KK$ respectively. 
Hence, by Lemma~\ref{lem:wempty}, $w_{\HH\KK} = w_{\KK\HH} = \varepsilon$.
This implies that either $\restrict{w}{1}\neq\vec{\varepsilon}$ or 
$\restrict{w}{2}\neq\vec{\varepsilon}$.  
Moroeover, $q_\HH$ and $q_\KK$ must also be final
in $M^1_\HH$ and $M^2_\KK$  respectively.
The rest of the thesis follows then by definition of orphan-message configuration.
\end{proof}

\begin{corollary}[Preservation of no orphan-message]
\label{prop:nomPreservation}
Let $S_1$ and $S_2$ be such that both  $RS(S_{1})$ and $RS(S_{2})$ do not contain any orphan-message configuration.
Then there is no orphan-message configuration in $RS(S)$.
\end{corollary}
\begin{proof}
By contradiction, let us assume there is an $s\in RS(S)$ which is an orphan-message configuration. We get
immediately a contradiction by Lemma \ref{lem:restrRSom}.
\end{proof}

\begin{proposition}[Preservation of no unspecified reception]
\label{prop:nurPreservation}
Let $S_1$ and $S_2$ be such that both  $RS(S_{1})$ and $RS(S_{2})$ do not contain any unspecified reception configuration.
Then there is no unspecified reception configuration in $RS(S)$.
\end{proposition}

\begin{proof}\,
[{\it Sketch}]
By contradiction, let us assume there is an $s= (\vec{q},\vec{w})\in RS(S)$ which is an unspecified reception configuration.
Moreover, let $\ttr \in \textbf{P}$ and $q_\ttr$  be the receiving state of $M_\ttr$ prevented from 
receiving any message from any of its buffers (Definition \ref{def:safeness}(\ref{def:safeness-ur})).
Without loss of generality, we assume $\ttr\in\roles_1$.
The following cases can occur:
\begin{description}
\item 
$q_\HH\not\in \widehat{Q_\HH}$.\\
By Lemma \ref{lem:nohatrestrict} we get $\restrict{s}{1}\in RS(S_1)$. 
Two sub-cases are possible:
\begin{description}
\item
$\ttr\neq\HH$\\
In such a case we get immediately a contradiction by  the hypothesis that $RS(S_{1})$ does not contain any unspecified reception configuration.
\item
$\ttr=\HH$\\
$q_\HH\,(= q_\ttr)$ is hence a receiving state. So let $\Set{(q_\HH,\tts_j\HH?a_j,\widehat{q_j})}_{j=1..m}$ be the set
of all the outgoing transitions from $q_\HH$ in $\delta_\HH$. 
By definition of unspecified reception configuration,  for any $j=1..m$, $\mid w_{\tts_j\HH}\mid > 0$
and $w_{\tts\HH}\not\in a_j\cdot\mathbb{A}^*$. 
By compatibility, and in particular by the absence of mixed states, we have just the following two possibilities:
    \begin{description}
\item
$\tts_j\neq\KK$ for any $j=1..m$.\\
By Fact \ref{fact:uniquesending}(\ref{fact:uniquesending-iii}) and definition of $\gateway{\cdot}$ we have that  
$$[(q_\HH,\tts_j\HH?a_j,\widehat{q_j})\in\delta_\HH  ~~\wedge ~~
\tts_j\neq\KK] \iff  
(q_\HH,\tts\HH?a_j,q_j)\in\delta^1_\HH$$
 This implies $\restrict{s}{1}$ to be an  unspecified reception configuration for $S_1$. Contradiction.
 \item
$\tts_j = \KK$ for any $j=1..m$.\\
Let $\restrictup{s}{\HH}$ and $\restrictup{s}{\KK}$ be defined as in the proofs of Lemmas
\ref{lem:alltrans} and \ref{lem:wempty}. We define now
\begin{itemize}
\item[-]
 $\restrictup{s}{\HH!}$ as the sequence made of the transition pairs in $\restrictup{s}{\HH}$
of the form  $\lts{(\_,\KK\HH?a,\_)}\lts{(\_,\HH\tts!a,\_)}$ with $\tts\neq\KK$,  and

\item[-]
$\restrictup{s}{\KK?}$ as the sequence made of the transition pairs in $\restrictup{s}{\KK}$
of the form  $\lts{(\_,\tts\KK?a,\_)}\lts{(\_,\KK\HH!a,\_)}$ with $\tts\neq\HH$.


\end{itemize}
Let now $n=\mid \restrictup{s}{\HH!} \mid$ and let $\restrictup{s}{\HH!}_{/n}$ be the sequence of the messages of the first $n$
elements of $\restrictup{s}{\HH!}$.
By compatibility, in particular $\lang{M^1_\HH}^\mC = \Dual{\lang{M^2_\KK}^\mC}$, it follows that,
if $w_{\KK\HH} \in  b\cdot\mathbb{A}^*$, then $\restrictup{s}{\KK?}_{/n+1}=\restrictup{s}{\KK?}_{/n}\cdot b$
and $\exists j. a_j = b$. So contradicting  that, for any $j=1..m$, 
$w_{\KK\HH}\not\in a_j\cdot\mathbb{A}^*$.
     \end{description}
 
\end{description}
\item 
$q_\HH=\widehat{q}\in \widehat{Q_\HH}$.\\ 
By Fact \ref{fact:uniquesending}(\ref{fact:uniquesending-i})
$q_\HH\in \widehat{Q_\HH}$ is a sending state such that $(q_\HH,\HH\tts!a,q''_\HH)\in{\delta}_\HH$. Hence it is impossible that $\ttr=\HH$.
So, let $\ttr\neq\HH$.
In such a case, by definition of $\gateway{\cdot}$, 
we have necessarily a unique transition of the form $(q'_\HH,\ttp\HH?a,q_\HH)\in{\delta}_\HH$.
Moreover, $q'_\HH\not\in \widehat{Q_\HH}$. So 
 there exists necessarily an element $s'\in RS(S)$ such that
$s'=(\vec{q'},\vec{w'})\lts{\ttp\HH?a}s$ with $q'_\HH\not\in \widehat{Q_\HH}$. It follows that also $s'$ is 
an unspecified-reception configuration and $\restrict{s'}{1} \in RS(S_1)$.
Then we get a contradiction by arguing like in the first case, sub-case $\ttr \neq \HH$.
\end{description}
\end{proof}

We are now ready to state our main results.

\begin{corollary}[Safety properties preservation]\hfill
\label{cor:safetypres}

\begin{enumerate}
\item
Let $S={S_{1}}\connect{\HH}{\KK} {S_{2}}$ be the system composed
from systems $S_1$ and $S_2$ with compatible CFSMs for the roles $\HH$ and $\KK$.
If $S_1$ and $S_2$ are both safe, then $S$ is safe.
\item
Let $\GTIR{\GG}{\mathbf{I}} =\GTIR{\GTIR{{\GG}_1}{\mathbf{H}}\connect{\HH}{\KK}\GTIR{{\GG}_2}{\mathbf{K}}}{\mathbf{I}}$
be the GTIR composed from GTIRs $\GTIR{\GG_1}{\mathbf{H}}$
and $\GTIR{\GG_2}{\mathbf{K}}$ via compatible interface roles $\HH$ and $\KK$.
If $\Sem{\GTIR{\GG_1}{\mathbf{H}}}$ and $\Sem{\GTIR{\GG_2}{\mathbf{K}}}$ are both safe,
then $\Sem{\GTIR{\GG}{\mathbf{I}}}$ is safe.
\end{enumerate}
\end{corollary}

As a consequence, by induction on the pairwise composition of GTIRs,
we obtain the following desired result.

\begin{corollary}
Let $\GTIR{\GG}{\mathbf{I}}$ be a GTIR  such that, for any global graph $G\in \components (\GTIR{G}{\mathbf{I}})$,\\
the system $(\projecton{G}{\ttp})_{\ttp\in\roles(G)}$ is safe.
Then $\Sem{\GTIR{\GG}{\mathbf{I}}}$ is safe.
\end{corollary}


\section{Conclusions}
\label{sect:conclusions}

We have proposed the GTIR formalism (Global Types with Interface Roles) to support the usage of global types
in the context of open systems,
whenever the underlying global type formalism,  like~\cite{DY12,TY15,TG18},  allows the interpretation of global
types in terms of systems of communicating finite state machines (CFSMs),
Our main result is that safety properties (deadlock-freeness, no orphan messages, no unspecified receptions)
are preserved by composing open systems when interface roles are compatible.

In \cite[Sect. 6]{GMY80}, the authors use the same compatibility notion for CFSMs showing that a system made of two CFSMs, which both are deterministic and do not have mixed states, is free from deadlocks and unspecified receptions.
The general aim of enhancing the expressive power of global type formalisms has been variously pursued in the literature.
For example, in \cite{Tanter14} the authors define a formalism where global types with initial and end points can be combined. 
The results in \cite{MY13} are slightly more related to our approach.  It is shown how to define choreographies partially specified, where only some processes are provided.
In case two choreographies are composable, some completely specified process of one
can be used instead of the unspecified ones in the other.

Even if some loose connections can be envisaged with the approach of {\em interface automata} of \cite{deAlfaro2001,deAlfaro2005},
our approach to open global types diverges from them in many relevant points:
First of all, an interface automaton describes the communication abilities of an automaton
with its environment in terms of input and output actions while internal behavior is described by internal actions.
GTIRs, however, emulate the expected behavior of the environment by providing distinguished interface roles
and their CFSMs, while internal behavior is modelled by the CFSMs of the other roles.
Interface automata
rely on synchronous communication while we consider asynchronous communication via FIFO buffers.
The crucial idea of compatibility for interface automata is that no error state should be reachable in the synchronous product
of two automata. An error state is a state, in which one automaton wants to send a message to the other but the other
automaton is not ready to accept it. This situation is related to unspecified reception in the asynchronous context.
The speciality of interface automata is, however, that an error state must be autonomously reachable,
i.e.\ without influence of the environment.
Since interface automata use synchronous message passing, the problem of orphans is empty.
Moreover, the theory of interface automata does not consider deadlock-freedom.
On the other hand, interface automata consider also refinement and preservation of compatibility by refinement.

In the future, we first want to study  whether our conditions for compatibility could be relaxed
still guaranteeing preservation of safety. 
Moreover, 
 it would be worth taking into account, besides safety properties, also liveness properties.
In particular,  the  generalised global types of  \cite{DY12}, 
at the cost of being less expressive than global types in \cite{TY15,TG18},
guarantee also liveness properties.
Properties preserved by connecting CFSM systems via {\em gateways} are worth to be investigated also
for  systems unrelated to global type formalisms.
For instance, a variety of communication properties are formalised for asynchronous I/O-transition systems in~\cite{HaddadHM13}.
Preservation by composition is shown there but using bags instead of FIFO buffers for communication.
Finally, the current composition operator for GTIRs is binary and thus can only lead to tree-like
compositions of global graphs. Therefore it would be challenging to see, how we could get
cyclic architectures.

\paragraph{Acknowledgements}
We are grateful to the anonymous referees for several helpful comments and suggestions.
We also thank Emilio Tuosto for some macros used to draw Figure \ref{fig:examplegg}.
The first author is also thankful to Mariangiola Dezani for her everlasting support.

%
\label{sect:bib}

\bibliographystyle{plainurl}
\bibliography{session}


\end{document}
